\documentclass[12pt]{article}
\usepackage{amsmath,amssymb,amsthm,graphicx,amssymb,amsfonts,amsbsy,amsmath,amsthm,mathtools,float,color,soul,hyperref,amsmath,amssymb,amsthm, pst-node,pst-text,pst-3d,amsfonts,pstricks,lineno,subfig,epstopdf}
\newtheorem{theorem}{Theorem}

\newtheorem{lemma}{Lemma}

\date{}
\input epsf.sty
\numberwithin{equation}{section}
\usepackage[skip=0pt]{caption}
\begin{document}\begin{center}
	\baselineskip .2in {\large\bf Interplay between reproduction and age selective harvesting delays of a single population non-autonomous system}
\end{center}
\begin{center}
	\baselineskip .2in {\bf N. S. N. V. K. Vyshnavi Devi$^{1}$, Debaldev Jana$^{1,\dag}$ and M. Lakshmanan$^{2}$ }
	
	{\small\it $^1$Department of Mathematics \& SRM Research Institute,}
	\\ SRM Institute of Science and Technology, Kattankulathur-603 203, Tamil Nadu, India. \\
	E-mail: vyshu$\textunderscore$nandi@ymail.com$^{1}$ and debaldevjana.jana@gmail.com$^{1,\dag}$

	{\small\it $^2$Centre for Nonlinear Dynamics,}
	\\ School of Physics, Bharathidasan University, Tiruchirapalli - 620 024, India.\\
	E-mail: lakshman.cnld@gmail.com$^{2}$
	\end{center}
\begin{abstract}
This paper is concerned with an analysis of the dynamics of a non-autonomous, single population age based growth model with harvesting formulation. First, we derive sufficient conditions for permanence and positive invariance. Then, by constructing a scalar function, namely the Lyapunov function, we arrive at a suitable criterion for global attractivity. With the help of Brouwer fixed point and continuation theorems, we obtain constraints for the existence of a positive periodic solution. Then we prove that there exists only one solution which is almost periodic in nature that is distinct from every other solution. Further, we carry out a numerical simulation to support the analytical findings.\\
    
    

\end{abstract}

\section{Introduction}\label{S1}
    Population dynamics speaks about the changes that occur in a given population as it varies with time, due to factors like birth, death, and migration. It acts as the fundamental principle for interpreting various patterns in fishery, predation, and harvesting. Population change in fishes is affected by many factors like availability of food, diseases, predators and climate \cite{MKS13, SA10, UG18}.\\
    
    We are particularly interested in \textit{Tenualosa ilisha} (Hilsa shad), a unique species of marine fish that migrates to inland waters in the Indian subcontinent in order to spawn \cite{DOF07}. Hilsa is usually found in the Bay of Bengal area which consists of large river systems \cite{DOFB08}. The length of the hilsa fish is strongly related to its body weight \cite{FMB15}. During the first year, the growth is roughly at the rate of an inch per month \cite{JSM51}. The spawning of Hilsa is seasonal \cite{BUS11} and is observed during “July-August to October-November” \cite{RMA17}. It takes 8-9 months for hilsa to grow to 163-225 mm in size. The fish attains maturity, typically at the onset of the second year or upon reaching the end of the first year when it undergoes reproduction \cite{BU15, SAV12}, generally when it is 250-300 mm in length \cite{RBT85}. Most commonly the hilsa is found to be 350-400 mm in length which weighs around 1 kg. It may grow up to 600 mm in size and weigh 3 kg at the age of four \cite{DOFB08}. By the end of its first and second years, it is usually of 16 and 27cm (approximately) in size, respectively \cite{MAR14}. Though biological studies have shown that hilsa can live up to 6.5 years of age, only a few fishes survive beyond 4 years \cite{HGC04}. Therefore, the lifespan of hilsa is considered to be 4 years \cite{HGC04}. 
    
    According to the data collected from upcountry waters, unrestricted fishing rights, water contamination, non-selective harvesting of hilsa (particularly juveniles), the destruction of spawning grounds lead to the decline in hilsa population \cite{DF08,DF12,HMS16}. If we continue harvesting hilsa before it attains sexual maturity, the population would face extinction in the near future \cite{JDR16}. Therefore it is essential to go for age-selective harvesting of Hilsa fishes.\\
    
    Think of a pine tree that has needle-like or scale-like leaves that are typically evergreen. It is well known that the photosynthetic ability of leaves changes as they mature and get old \cite{BB00}. Photosynthesis takes place at a low pace in young leaves, maximum pace in matured leaves, and the rate decreases as leaves age \cite{FRO52,SBN35}. This is just one example which shows the variation in parametric influence at young stand ages and old ages \cite{SXB12}. This kind of time-dependent variation is often seen in ecological growth models that are influenced by environmental fluctuations. For an intricate ecological growth model, it is of utmost importance to identify these influential, time-variant parameters and their impact on the system dynamics. The same can be applied to the growth model in hilsa fish.\\
    
    As discussed earlier, for biological and economic benefits, we can harvest the hilsa when it crosses its maturity age and harvestable yield. By this, we can also assure that the population will be healthy and sustainable. The rate of harvest can be measured in catch per man-hour (rate of catch) \cite{EW46}. Several authors \cite{EW46,MB47} mentioned that temperature of water influences the rate of catch, since “water temperature is a major spawning stimulus” \cite{SJ59}. This shows the dependence of birth rate on water temperature. Also, water level, turbidity, heavy rains or lack of frequent rains, air temperature, etc are some of the factors affecting the system parameters which in turn affect the rate of harvest \cite{SJ59}. Some physical and chemical/biochemical factors like water temperature, pH, alkalinity, hardness, etc are important to fish growth, mortality and harvestable yield \cite{VRC05}. All the above-mentioned factors are time variant. They influence the system parameters, thereby suggesting the time dependence of the rate of harvest and harvestable yield.\\
    
    The standard population equilibrium model, namely Verhulst's growth model \cite{V38}, is a typical application of logistic equation which describes population growth in a limited environment. Many illustrious people like Cunningham \cite{CW57},  Wright \cite{WE55}, Kakutani and Markus \cite{KM58}, Oster and Takahashi \cite{OT74}, Auslander et al. \cite{AOH74}, Kaplan and Yorke \cite{KY75}, Beddington and May \cite{BM75}, Blythe et al. \cite{BNG82}, Seifert \cite{S87},Kuang \cite{K93}, Walther \cite{W95}, Rodriguez \cite{R98} and Ruan \cite{RS06} worked towards developing an appropriate logistic delay differential equation (DDE), leading to single species growth and age structured DDE model. According to Arino, Wang and Wolkowicz \cite{AWW06} the rate at which population changes depends on the following aspects: growth (includes delay $\vartheta$), death and intraspecific competition. It is also assumed that the death rate and intraspecific competition rate both together contribute to the rate of decline, which is instantaneous. The death rate is represented using a linear term and a quadratic term is used to describe the rate of intraspecific competition \cite{JDA16}. Now in our model we will assume that individuals of hilsa fish which are born before a time $\vartheta_1$ from the present instant $t$ are sexually able to give rise to new offspring and we can harvest for our economic profit those hilsa fishes which are born before $\vartheta_2$ time from now ($t$). So, we can state that a density $\digamma(t)$ of hilsa fish given by the expression $\digamma(t)=\frac{\varrho\varphi \digamma(t-\vartheta_1)}{\varphi e^{\varphi \vartheta_1}+c (e^{\varphi \vartheta_1}-1)\digamma(t-\vartheta_1)}$ (see Appendix A) can give birth at time $t$ and we can further harvest a density $\digamma(t)=\frac{qE\varphi \digamma(t-\vartheta_2)}{\varphi e^{\varphi \vartheta_2}+c (e^{\varphi \vartheta_2}-1)\digamma(t-\vartheta_2)}$ (see Appendix A) at time $t$. So our model of hilsa fish becomes \cite{JDS19}:
\begin{equation}\label{eq:q}
    \digamma^{'}(t)=\dfrac{\varrho\varphi \digamma(t-\vartheta_1)}{\varphi e^{\varphi\vartheta_1}+c(e^{\varphi\vartheta_1}-1)\digamma(t-\vartheta_1)}-\varphi \digamma(t)-c\digamma^2(t)-\dfrac{qE\varphi \digamma(t-\vartheta_2)}{\varphi e^{\varphi\vartheta_2}+c(e^{\varphi\vartheta_2}-1)\digamma(t-\vartheta_2)}\;,
\end{equation}\\
    where $\digamma(t)$ specifies the count of fishes at $t$(time), prime stands for differentiation with respect to time, $\varrho$ represents the birth rate, and $\varphi$ is natural mortality. Also $c$, $E$, and $q$ indicate removal due to intraspecific competition, fishing effort, and catchability coefficient, respectively. Various environmental factors and human aspects have an effect on the value of $q$ \cite{DT17}. Time-dependent catchability is very common and is caused by environmental, anthropogenic and biological processes \cite{WMJ09}. Note that in the above $\varrho$ represents the fishes reaching a certain size or reproductive stage. $\varphi$ includes non-human predation, diseases and old age. $E$ includes particulars like the number of individuals involved in fishing, the number of boats available, etc. Intraspecific competition $c$ varies in line with resource attributes and grouping with respect to time and space, environmental factors, etc \cite{WWH06}. Also, the size of the fish population affects the number of fishes caught (harvesting) and vice versa, which clearly describes the time dependence of every parameter.\\
    
    Jana, Dutta and Samanta \cite{JDS19} carried out a case study on the dynamics of model \eqref{eq:q} wherein the parameters are time-independent (autonomous system) at Sundarban estuary, Bay of Bengal, India. For one particular set of parameters (constant values), if the autonomous system is stable, for a similar parameter set which is time variant (non-autonomous system), the stability properties might not be the same. This has motivated us to consider the time-variant environment at Bay of Bengal which in turn affects the life of hilsa. Thus, incorporating the time factor into the parameters of the model \eqref{eq:q} we get a non-autonomous system as follows,
\begin{equation}\label{eq:1}
\begin{split}
   \digamma^{'}(t)=\dfrac{\varrho(t)\varphi(t)\digamma(t-\vartheta_1)}{\varphi(t)e^{\varphi(t)\vartheta_1}+c(t)(e^{\varphi(t)\vartheta_1}-1)\digamma(t-\vartheta_1)}-\varphi(t)\digamma(t)-c(t)\digamma^2(t) \\
-\dfrac{q(t)E(t)\varphi(t)\digamma(t-\vartheta_2)}{\varphi(t)e^{\varphi(t)\vartheta_2}+c(t)(e^{\varphi(t)\vartheta_2}-1)\digamma(t-\vartheta_2)}\cdot
\end{split}
\end{equation}
   Many authors like Li and Takeuchi \cite{HY15}, Chen, Chen and Shi \cite{CS08}, Cui and Takeuchi \cite{CT06}, Fan and Kaung \cite{FK04}, Zeng and Fan \cite{ZF08} have discussed non-autonomous dynamical systems, studied their behavior and obtained rich dynamics. For the above age-based growth model \eqref{eq:1} with time-variant parameters, our objective here is to study the dynamical aspects like positive invariance, permanence, and global attractivity by formulating a Lyapunov function. Also, we wish to verify the occurrence of a positive solution that is periodic in nature and a unique almost periodic solution.\\
   
   We proceed in a systematic manner starting with some basic definitions in section 2. Then in sections 3, 4 and 5, we demonstrate positive invariance, permanence and global attractivity for system \eqref{eq:1}, respectively. Then comes a set of proofs in sections 6 and 7 for the occurrence of positive solutions that are periodic and (unique) almost periodic in nature to system \eqref{eq:1}, respectively. This is followed by many numerical examples and their graphical interpretation in section 8. The paper ends with a conclusion in section 9.
   
\section{Basic definitions}\label{s2}
   Some notations, assumptions and definitions used in the following sections are introduced here in this section.\\
\textbullet {} Only those solutions $\digamma(t)$ with $\digamma(t_0)>0$ $\forall$ $t_0>0$ are taken into consideration (for biological reasons). \\
\textbullet {} $\varrho(t)$, $c(t)$, $E(t)$, $\varphi(t)$ and $q(t)$ have both lower and upper bounds (positive constants) and are assumed to be continuous. \\
\textbullet {} For a bounded (enclosed within positive values) continuous function $a(t) \in$ $\mathbb{R}$, $a_L$ is the infimum of $a(t)$ and $a_M$ is the supremum of $a(t)$ where $t$(time) is a value in $\mathbb{R}$.\\
\textbullet {} Every co-efficient of system \eqref{eq:1} must satisfy
\begin{center}
min$\{\varrho_L,\varphi_L,q_L,c_L,E_L\}>0$ \&
max$\{\varrho_M,\varphi_M,q_M,c_M,E_M\}<\infty$.
\end{center} \vspace{0.5cm}
\textbf{Definition 2.1.} The interval $[\digamma^\ast-\varepsilon,\digamma^\ast+\varepsilon]$ with $\digamma^\ast$ lying in it, such that any solution $\digamma(t)$ of system \eqref{eq:1}, which is lying in the interval at $t=t_0$, will remain in this interval $\forall$ $t\geqslant t_0$ and will tend to $\digamma^\ast$ as $t\to\infty$. This interval with reference to system \eqref{eq:1} is positively invariant. \vspace{0.5cm}\\
\textbf{Definition 2.2.} System \eqref{eq:1} shall remain permanent whenever $\exists$ stationary values $\Delta_1$ and $\Delta_2$ (both positive) with $0<\Delta_1\leqslant\Delta_2$ $\ni$ $\lim\limits_{t\to +\infty}$ inf $\digamma(t)\geqslant\Delta_1$ and $\lim\limits_{t\to +\infty}$ sup $\digamma(t)\leqslant\Delta_2$ holding good for those solutions of system \eqref{eq:1} that have positive initial values. \vspace{0.5cm}\\
\textbf{Definition 2.3.} If $\lim\limits_{t\to\infty}$ $|\digamma(t)-\digamma^\ast(t)|=0$ holds for a solution $\digamma(t)$ to system \eqref{eq:1} then $\digamma^\ast(t)$, which is a bounded non-negative solution to system \eqref{eq:1}, shall be globally attractive. \vspace{0.2cm}\\
\textbf{Definition 2.4.} The set of all solutions of system \eqref{eq:1} is ultimately bounded when there is a $V>0$ $\ni$ for every solution $\digamma(t)$ of \eqref{eq:1}, $\exists$ $W>0$ $\ni$ $||\digamma(t)||\leqslant V$ for all $t\geqslant t_0 + W.$ \vspace{0.5cm}\\
\textbf{Definition 2.5.} When $C(A)$ forms a collection of continuous functions over a compact metric space $A$ and $j \in B\subseteq C(A)$, B can be called an equi-continuous family of functions if for every $ \varepsilon>0$ $\exists$ $\Delta>0\ni$ $d(x,y)<\Delta$ as long as $|j(x)-j(y)|<\varepsilon$ $\forall$ $j\in B$. 
\section{Positive Invariance}\label{S3}
\begin{theorem}\label{thm:t1}
If $ \varrho_L\varphi_L > (c_M)^2(M^\varepsilon)^2 $ then 
\begin{center}
$ K_\varepsilon=\{\digamma(t)\in \mathbb{R} / m^\varepsilon \leqslant \digamma(t) \leqslant M^\varepsilon \} $
\end{center}
will be a positively invariant set
with reference to system \eqref{eq:1}. Here $ M^\varepsilon =\dfrac{\varrho_M\varphi_M}{c_L}+\varepsilon $ and
$ m^\varepsilon=\dfrac{\varrho_L\varphi_L-(c_M)^2(M^\varepsilon)^2 }{c_M} -\varepsilon $ 
and $ \varepsilon \geqslant 0 $ is sufficiently small $\ni$ $ m^\varepsilon > 0. $
\end{theorem}
\begin{proof}
Let us say that $ \digamma(t) $ (with initial condition $ \digamma(t_0)>0) $ is a solution to system \eqref{eq:1} with $ m^\varepsilon \leqslant \digamma(t_0) \leqslant M^\varepsilon $. That is, $ \digamma(t_0) \in K_\varepsilon $. From system \eqref{eq:1}, when $ t \geqslant t_0 $ we have
\begin{center}
$ \digamma^{'}(t) \leqslant \varrho_M\varphi_M\digamma(t-\vartheta_1) - \varphi_L \digamma(t) - c_L\digamma^2 (t) $. 
\end{center}
Rearranging the terms, we get, 
\begin{center}
$ \dfrac{\varrho_M\varphi_M\digamma(t-\vartheta_1)}{\digamma(t)} - \dfrac{\digamma^{'}(t)}{\digamma(t)} \geqslant \varphi_L + c_L\digamma(t) $.
\end{center}
$\digamma^{'}(t)$, rate of change in hilsa population, is considered to be a very small positive value (for biological reasons).
This implies that for every $ t\geqslant t_0 $, 
\begin{center}
$ \dfrac{\varrho_M\varphi_M\digamma(t-\vartheta_1)}{\digamma(t)} \geqslant \varphi_L + c_L\digamma(t) $.
\end{center}
It means that
\begin{center}
$ \dfrac{\varrho_M\varphi_M \digamma(t-\vartheta_1)}{c_L \digamma(t)} \geqslant \dfrac{\varphi_L}{c_L}+\digamma(t).  $
\end{center}
Assuming $\dfrac{\digamma(t-\vartheta_1)}{\digamma(t)} \approx 1$, 
\begin{center}
$ \dfrac{\varrho_M\varphi_M}{c_L} \geqslant \dfrac{\varphi_L}{c_L}+\digamma(t).$ 
\end{center}
Therefore, we can say that
\begin{center}
$ \dfrac{\varrho_M\varphi_M}{c_L} \geqslant \digamma(t). $ 
\end{center}
For every $\varepsilon >0 $ we write 
\begin{center}
$ \digamma(t)\leqslant\dfrac{\varrho_M\varphi_M}{c_L} + \varepsilon=M^\varepsilon.  $
\end{center}
From system \eqref{eq:1}, when $ t\geqslant t_0 $ we obtain
\begin{center}
$ \digamma^{'}(t)\geqslant\dfrac{\varrho_L\varphi_L}{c_M} - \varphi_M \digamma(t) - c_M(M^\varepsilon)^2. $
\end{center}
This is same as
\begin{center}
$ \digamma(t) \geqslant\dfrac{\varrho_L\varphi_L - (c_M)^2(M^\varepsilon)^2 }{c_M\varphi_M} - \dfrac{\digamma^{'}(t)}{\varphi_M}\cdot $
\end{center}
We can also say that
\begin{center}
$ \digamma(t) \geqslant \dfrac{\varrho_L\varphi_L - (c_M)^2(M^\varepsilon)^2 }{c_M\varphi_M}\cdot $
\end{center}
And hence we get
\begin{center}
$ \digamma(t) \geqslant \dfrac{\varrho_L\varphi_L - (c_M)^2(M^\varepsilon)^2 }{c_M}=m_\varepsilon + \varepsilon$. 
\end{center}
We write ($ \varepsilon>0 $)
\begin{center}
$ \digamma(t) \geqslant m_\varepsilon. $ 
\end{center}
Hence $ m^\varepsilon \leqslant \digamma(t) \leqslant M^\varepsilon $ $\forall$ $ t \geqslant$ initial time ($t_0$).
This means that, $ \digamma(t) \in K_\varepsilon $ for every $ t \geqslant$ initial time ($t_0 $). Hence $ K_\varepsilon $ satisfies positive invariance property with reference to system \eqref{eq:1}. 
\end{proof}
\section{Permanence}\label{S4}
\begin{theorem}\label{thm:t2}
If $ \varrho_L\varphi_L > (c_M)^2(M^0)^2 $ condition is met we can say that system \eqref{eq:1} is \vspace{0.2cm} \\permanent, where  $M^0=\dfrac{\varrho_M\varphi_M}{c_L}$, $ m^0= \dfrac{\varrho_L\varphi_L -(c_M)^2(M^0)^2}{c_M}\cdot$
\end{theorem}
\begin{proof}
As seen above, system \eqref{eq:1} gives 
\begin{center}
$ \digamma^{'}(t)>\dfrac{\varrho_L\varphi_L}{c_M}- \varphi_M \digamma(t)-c_M(M^\varepsilon)^2. $
\end{center}
When $ \dfrac{\varrho_L\varphi_L}{c_M}>c_M(M^0)^2 $ we get,
\begin{center}
\vspace{-0.5cm}
$ \lim\limits_{t\rightarrow+\infty} $ inf $ \digamma(t) \geqslant \dfrac{\varrho_L\varphi_L -(c_M)^2(M^0)^2}{c_M}= m^0. $
\end{center}
Also, system \eqref{eq:1} implies
\begin{center}
$ \digamma^{'}(t) < \varrho_M\varphi_M \digamma(t-\vartheta_1)-\varphi_L \digamma(t)-c_L\digamma^2(t) $,
\end{center}
which implies that 
\begin{center}
\vspace{-0.5cm}
$ \lim\limits_{t\rightarrow+\infty} $ sup $ \digamma(t) \leqslant \dfrac{\varrho_M\varphi_M}{c_L}=M^0. $
\end{center}
All the conditions in the definition of permanence are satisfied. That is, population of hilsa never becomes extinct. Hence system \eqref{eq:1} is going to be permanent whenever $ \varrho_L\varphi_L > (c_M)^2(M^0)^2 $. 
\end{proof}
\section{Global Attractivity}\label{S5}
\begin{theorem}\label{thm:t3}
Let us consider $ \digamma^\ast(t) $ to be a bounded (enclosed within positive values) positive ($>0$) solution to system \eqref{eq:1}. Whenever
$ \varrho_L\varphi_L > (c_M)^2(M^\varepsilon)^2 $ and\\
$M=$
\begin{equation}\label{eq:2}
\begin{split}
\frac{-\varrho(t)\varphi^2(t)e^{\varphi(t)\vartheta_1}}{(\varphi(t)e^{\varphi(t)\vartheta_1}+c(t)(e^{\varphi(t)\vartheta_1}-1)M^\varepsilon)^2} + 4c(t)m^\varepsilon + 2\varphi(t)
+ \frac{q(t)E(t)\varphi^2(t)e^{\varphi(t)\vartheta_2}}{(\varphi(t)e^{\varphi(t)\vartheta_2}+c(t)(e^{\varphi(t)\vartheta_2}-1)M^\varepsilon)^2} > 0
\end{split}
\end{equation}
hold, then $ \digamma^\ast(t) $ is globally attractive.
\end{theorem}
\begin{proof}
Consider $ \digamma(t) $ to be any other solution to system \eqref{eq:1}. Since $ K_\varepsilon $ is ultimately bounded
region of this system, $ \exists $ $ W $ \textgreater $ 0 $ $\ni$ $ \digamma(t)\in K_\varepsilon $ and $ \digamma^\ast(t)\in K_\varepsilon $ $\forall $ $t\geqslant t_0+W . $ \\
Consider the function
\begin{center}
$ U(t)=\frac{1}{2}(\digamma(t)-\digamma^\ast(t))^2$.
\end{center}
Let us consider \vspace{0.2cm}\\
$ \varphi(t,t-\vartheta_1,\digamma(t-\vartheta_1),\digamma^\ast(t-\vartheta_1))= $ 
\begin{center}
$ [\varphi(t)e^{\varphi(t)\vartheta_1}+c(t)(e^{\varphi(t)\vartheta_1}-1)\digamma(t-\vartheta_1)][\varphi(t)e^{\varphi(t)\vartheta_1}+c(t)(e^{\varphi(t)\vartheta_1}-1)\digamma^\ast(t-\vartheta_1)] $,
\end{center}
$ \varphi(t,t-\vartheta_2,\digamma(t-\vartheta_2),\digamma^\ast(t-\vartheta_2))= $ 
\begin{center}
$ [\varphi(t)e^{\varphi(t)\vartheta_2}+c(t)(e^{\varphi(t)\vartheta_2}-1)\digamma(t-\vartheta_2)][\varphi(t)e^{\varphi(t)\vartheta_2}+c(t)(e^{\varphi(t)\vartheta_2}-1)\digamma^\ast(t-\vartheta_2)] $.
\end{center}
The differentiation of $ U(t) $ with respect to time $t$ through system \eqref{eq:1} is given by
\begin{center}
$ U'(t)=[\digamma(t)-\digamma^\ast(t)]\bigg[\dfrac{\varrho(t)\varphi(t)\digamma(t-\vartheta_1)}{\varphi(t)e^{\varphi(t)\vartheta_1}+c(t)(e^{\varphi(t)\vartheta_1}-1)\digamma(t-\vartheta_1)}-\varphi(t)\digamma(t) $ \vspace{0.25cm}\\
$ -c(t)\digamma^2(t)-\dfrac{q(t)E(t)\varphi(t)\digamma(t-\vartheta_2)}{\varphi(t)e^{\varphi(t)\vartheta_2}+c(t)(e^{\varphi(t)\vartheta_2}-1)\digamma(t-\vartheta_2)} $ \vspace{0.25cm}\\
$ -\dfrac{\varrho(t)\varphi(t)\digamma^\ast(t-\vartheta_1)}{\varphi(t)e^{\varphi(t)\vartheta_1}+c(t)(e^{\varphi(t)\vartheta_1}-1)\digamma^\ast(t-\vartheta_1)}+\varphi(t)\digamma^\ast(t) $ \vspace{0.25cm}\\
$ +c(t)\digamma^{\ast2}(t)+ \dfrac{q(t)E(t)\varphi(t)\digamma^\ast(t-\vartheta_2)}{\varphi(t)e^{\varphi(t)\vartheta_2}+c(t)(e^{\varphi(t)\vartheta_2}-1)\digamma^\ast(t-\vartheta_2)}\bigg] $.
\end{center}
The above implies 
\begin{center}
$ U'(t)=[\dfrac{\varrho(t)\varphi^2(t)e^{\varphi(t)\vartheta_1}}{\varphi(t,t-\vartheta_1,\digamma(t-\vartheta_1),\digamma^\ast(t-\vartheta_1))}][\digamma(t)-\digamma^\ast(t)][\digamma(t-\vartheta_1)-\digamma^\ast(t-\vartheta_1)] $
\vspace{0.2cm}\end{center}
\begin{center}
$ - c(t)[\digamma(t)+ \digamma^\ast(t)][\digamma(t)-\digamma^\ast(t)]^2-\varphi(t)[\digamma(t)-\digamma^\ast(t)]^2 $
\vspace{0.2cm}\end{center}
\begin{center}
$ -\dfrac{q(t)E(t)\varphi^2(t)e^{\varphi(t)\vartheta_2}}{\varphi(t,t-\vartheta_2,\digamma(t-\vartheta_2),\digamma^\ast(t-\vartheta_2))} [\digamma(t)-\digamma^\ast(t)][\digamma(t-\vartheta_2)-\digamma^\ast(t-\vartheta_2)] $.
\vspace{0.5cm}\end{center}
Using $ fg\leqslant\frac{1}{2}(f^2+g^2)$,
\begin{center}
$ U'(t)\leqslant - \dfrac{1}{2}[\dfrac{-\varrho(t)\varphi^2(t)e^{\varphi(t)\vartheta_1}}{(\varphi(t)e^{\varphi(t)\vartheta_1}+c(t)(e^{\varphi(t)\vartheta_1}-1)M^\varepsilon)^2}+4c(t)m^\varepsilon+2\varphi(t) $ \vspace{0.5cm}\\ 
$ +\dfrac{q(t)E(t)\varphi^2(t)e^{\varphi(t)\vartheta_2}}{(\varphi(t)e^{\varphi(t)\vartheta_2}+c(t)(e^{\varphi(t)\vartheta_2}-1)M^\varepsilon)^2}][\digamma(t)-\digamma^\ast(t)]^2 $.
\vspace{0.5cm}\end{center}
By the definition of $ U(t) $ and the condition \eqref{eq:2} we get,
\begin{center}
$ U'(t)\leqslant -M U(t) $. 
\end{center}
Integrating the above inequality we have
\begin{center}
$ \mbox{log}U(t)\leqslant-Mt $.
\end{center}
Since $ U(t)$ is positive, we get
\begin{center}
$ \lim\limits_{t\rightarrow+\infty}U(t)=0 $.
\end{center}
That implies
\begin{center}
$ \lim\limits_{t\rightarrow+\infty}|\digamma(t)-\digamma^\ast(t)|=0 $.
\end{center}
Hence $ \digamma^\ast(t) $ happens to be globally attractive.
\end{proof}
\section{Existence of periodic solution}\label{S6}
Due to seasonal factors like mating habits, weather conditions, scope of getting food, etc. every parameter in the system \eqref{eq:1} is periodic of some common period. Thus, incorporating a periodic environment to system \eqref{eq:1} we consider every parameter in the system \eqref{eq:1} to be $ \Theta $ periodic in $ t$; $ \varrho(t+\Theta)=\varrho(t) $, $ \varphi(t+\Theta)=\varphi(t) $, $ q(t+\Theta)=q(t) $, $ c(t+\Theta)=c(t) $, $ E(t+\Theta)=E(t) $. \vspace{0.05cm}\\ 
For a continuous, $\Theta$-periodic function, denote $ \Tilde{k}=\frac{1}{\Theta}\int_{0}^{\Theta}k(t)dt $. \vspace{0.05cm}\\
We make use of the following Lemmas \ref{lem:1} and \ref{lem:2} to prove Theorems \ref{thm:t4} and \ref{thm:t5} respectively.
\begin{lemma}\label{lem:1}
(Brouwer fixed point theorem)\normalfont{\cite{BU11}} \textit{For a continuous mapping $ \kappa$ : $ \Bar{\Omega} $ $\rightarrow$  $ \Bar{\Omega} $ where $ \Bar{\Omega} $ is a bounded closed convex set in $\mathbb{R}^n$, then $\Bar{\Omega}$ contains a $ \digamma^\ast $ in such a way that $ \kappa(\digamma^\ast)=\digamma^\ast$.}
\end{lemma}
\begin{theorem}\label{thm:t4}
If $ \varrho_L\varphi_L > (c_M)^2(M^\varepsilon)^2 $ holds, there exists a minimum of one positive($>0$) $\Theta$-periodic solution, $\digamma^\ast(t)$, to system \eqref{eq:1} where $\digamma^\ast(t) \in K_\varepsilon $.
\end{theorem} 
\begin{proof}
Construct a mapping $ \kappa:\mathbb{R}\rightarrow\mathbb{R} $ by
$ \kappa(\digamma_0)=\digamma(t_0+\Theta,t_0,\digamma_0), \digamma_0\in\mathbb{R} $
and $ \digamma(t,t_0,\digamma_0) $ is a solution to \eqref{eq:1} through the point $ (t_0,\digamma_0) . $
We know that the set $ K_\varepsilon $ is positive invariant w.r.t system \eqref{eq:1}.
Hence $ \kappa $ maps $ K_\varepsilon $ into $ K_\varepsilon $. Hence, we can say that $ \kappa(K_\varepsilon)$ is a subset of $K_\varepsilon . $
Owing to the continuity of the solution to system \eqref{eq:1} with reference to the initial point, $\kappa $ is continuous. 
Also, $ K_\varepsilon $ satisfies the definitions of a bounded set, closed set and convex set in $ \mathbb{R}. $
Using Lemma \ref{lem:1}, $ K_\varepsilon $ contains minimum one fixed point in $ \kappa $.
This means that $\exists$ $ \digamma^\ast\in K_\varepsilon $ $\ni$ 
$ \kappa(\digamma^\ast)=\digamma^\ast . $
We know that $ \kappa(\digamma^\ast)=\digamma(t_0+\Theta,t_0,\digamma^\ast) $
which implies $ \digamma^\ast=\digamma(t_0+\Theta,t_0,\digamma^\ast) . $ 
That is, there exists atleast one positive periodic solution $ \digamma^\ast $ and
invariance of $ K_\varepsilon $ assures that $ \digamma^\ast\in K_\varepsilon . $
\end{proof}
To derive Theorem \ref{thm:t5}, we make use of some facts about the degree theory: \\
Consider two normed vector spaces $X$ and $Z$. Consider a linear operator $L$: Domain $L \rightarrow Z$ where Domain $L$ is a subset of $X$ \cite{ZX15}. Let $S$ be any operator mapping from the normed vector space X to normed vector space Z in such a way that $S$ is continuous. Suppose dimension of Kernel $L$ and co-dimension of Image $L$ are equal and finite and if Image $L$ is a closed subset of $Z$ then $L$ becomes a zero index Fredholm mapping. If projections $P$ and $Q$ mapping from $X$ to $X$ and $Z$ to $Z$, respectively, are continuous and maintain equality between sets Image $P$ and Kernel $L$ and sets Kernel $Q$, Image $L$ and Image $(I-Q)$, and suppose $L$ is a zero index Fredholm mapping, then the mapping $L|$Domain $L$ $\cap$ Kernel $P$ from the set $(I-P)X$ to Image $L$ will be an invertible mapping, represented with the notation $K_P$. Over the set $\Bar{\Omega}$, $S$ is known to be $L$-compact if $K_P(I-Q)S$ mapping from $\Bar{\Omega}$ to the normed vector space $X$ becomes compact and $QS(\Bar{\Omega})$ remains bounded (only when a subset of $X$, $\Omega$ remains open and bounded in $X$). $J$ : Image $Q$ $\rightarrow$ Kernel $L$, an isomorphism, exists because, Image $Q$ is isomorphic to Kernel $L$. 
\begin{lemma}\label{lem:2}
(Continuation theorem)\normalfont{\cite{GM77}} \textit{For a zero index Fredholm mapping L and a L-compact set on $ \Bar{\Omega},$ denoted by S, if \\
(a) Any solution $x$ to $ Lx=\Lambda Sx $ is $\ni$ $ x\notin\partial\Omega$ ($\forall$ $ \Lambda \ni $ $0< \Lambda <1 $), \\
(b) For every $ x$ lying in the set $ Kernel L \cap \partial\Omega $, $ QSx$ is never zero. Also, the Brouwer degree, given by $ deg\{JQS,\Omega\cap Kernel L,0\} $ never equals zero.\\
If the above two conditions hold, equation $ Lx=Sx $ contains minimum one solution inside the set $\Bar{\Omega} \cap Domain L $.}
\end{lemma}
\begin{theorem}\label{thm:t5}
If $ \Tilde{\varrho\varphi}>(\Tilde{c} M^\varepsilon)^2 $ holds, system \eqref{eq:1} must possess minimum one positive($>0$) $\Theta$-periodic solution.
\end{theorem}
\begin{proof}
Let $ \digamma(t)=\exp\{n(t)\}, \digamma(t-\vartheta_1)=\exp\{n(t-\vartheta_1)\}, \digamma(t-\vartheta_2)=\exp\{n(t-\vartheta_2)\} $. \\
Then system \eqref{eq:1} becomes \vspace{-0.5cm}
\begin{equation}\label{eq:z}
\begin{split}
n'(t)=\frac{\varrho(t)\varphi(t)e^{n(t-\vartheta_1)}}{e^{n(t)}[\varphi(t)e^{\varphi(t)\vartheta_1}+c(t)(e^{\varphi(t)\vartheta_1}-1)e^{n(t-\vartheta_1)}]}-\varphi(t)-c(t)e^{n(t)} \\ -\frac{q(t)E(t)\varphi(t)e^{n(t-\vartheta_2)}}{e^{n(t)}[\varphi(t)e^{\varphi(t)\vartheta_2}+c(t)(e^{\varphi(t)\vartheta_2}-1)e^{n(t-\vartheta_2)}]}\cdot 
\end{split}
\end{equation}
Consider 
\begin{center}
$ X=\{n\in C(\mathbb{R},\mathbb{R}):\digamma(t+\Theta)=\digamma(t)\} $ and 
$ \| n \|= $ $\max\limits_{t\in[0,\Theta]}|n(t)| $ for $ n \in X $.
\end{center}
$X$ equipped with the mentioned norm $ \| \, \| $ becomes a complete normed space.\\
Suppose
\begin{center} \vspace{-0.5cm}
$ Sn=S(t)=\dfrac{\varrho(t)\varphi(t)e^{n(t-\vartheta_1)}}{e^{n(t)}[\varphi(t)e^{\varphi(t)\vartheta_1}+c(t)(e^{\varphi(t)\vartheta_1}-1)e^{n(t-\vartheta_1)}]}-\varphi(t)-c(t)e^{n(t)} $ \vspace{0.1cm}\\ 
$ -\dfrac{q(t)E(t)\varphi(t)e^{n(t-\vartheta_2)}}{e^{n(t)}[\varphi(t)e^{\varphi(t)\vartheta_2}+c(t)(e^{\varphi(t)\vartheta_2}-1)e^{n(t-\vartheta_2)}]} $.
\end{center}
We know that $L$ is a linear map, $L$: Domain $L \rightarrow Z$ where Domain $L$ is a subset of $X$. Let $ L(n)=n' $ where $ n\in X$ and $ n'  \in Z$. We also know that projections $P$ and $Q$ mapping from $X$ to $X$ and $Z$ to $Z$, respectively, are continuous. Let $ P(n)= Q(n)= \frac{1}{\Theta}\int_{0}^{\Theta}n(t)dt$ for $n$ that lies in $X $.
Then \vspace{-0.5cm}
\begin{center}
Kernel $ L=\{n\in X/ n(t)=r\in\mathbb{R} $ for $ t\in\mathbb{R}\}, $ \vspace{0.1cm}\\
Image $ L=\{n\in X/\int_{0}^{\Theta}n(t)dt=0\}, $
\end{center}
where $ r \in\mathbb{R}, $ dimension Kernel $ L$ and co-dimension of Image $ L $ are equal to one and thus finite ($< \infty$), also Image $ L $ $\in Z$ is a closed set. 
Hence L becomes a zero index Fredholm mapping.
Also, the sets Image $ L $, Image $ (I-Q)$ and Kernel $ Q$ maintain equality. The same holds true for sets Image $ P$ and Kernel $ L $. 
Mapping $ K_P$ defined by $ K_P$: Image $ L\rightarrow $ Domain $ L\cap $ Kernel $ P $ is given by
\begin{center}
$ K_P(n)=\int_{0}^{t}n(s)ds-\frac{1}{\Theta}\int_{0}^{\Theta}\int_{0}^{t}n(s)dsdt $.
\end{center}
Over the set $\Bar{\Omega}$, $S$ is known to be $L$-compact because the sets $K_P(I-Q)S$ and $QS$ are continuous (only when a subset of $X$, $\Omega$ remains open and bounded in $X$).
Then the equation $ L\digamma =\Lambda S\digamma $ where $0< \Lambda <1 $ gives
\begin{center}
$ n'(t)=\Lambda\bigg[\dfrac{\varrho(t)\varphi(t)e^{n(t-\vartheta_1)}}{e^{n(t)}[\varphi(t)e^{\varphi(t)\vartheta_1}+c(t)(e^{\varphi(t)\vartheta_1}-1)e^{n(t-\vartheta_1)}]}-\varphi(t)-c(t)e^{n(t)} $ \\ 
$ -\dfrac{q(t)E(t)\varphi(t)e^{n(t-\vartheta_2)}}{e^{n(t)}[\varphi(t)e^{\varphi(t)\vartheta_2}+c(t)(e^{\varphi(t)\vartheta_2}-1)e^{n(t-\vartheta_2)}]}\bigg]\cdot $
\end{center}
If $ n(t)\in X $ is any arbitrary solution to the above equation for $ \Lambda\in (0,1), $ \\
we get, by integration of above equation over $ [0,\Theta]$,
\begin{center}
$ \int_{0}^{\Theta}n'(t)dt=\Lambda\bigg[\int_{0}^{\Theta}\dfrac{\varrho(t)\varphi(t)e^{n(t-\vartheta_1)}}{e^{n(t)}[\varphi(t)e^{\varphi(t)\vartheta_1}+c(t)(e^{\varphi(t)\vartheta_1}-1)e^{n(t-\vartheta_1)}]}dt-\int_{0}^{\Theta}\varphi(t)dt-\int_{0}^{\Theta}c(t)e^{n(t)}dt $ \\
$ -\int_{0}^{\Theta}\dfrac{q(t)E(t)\varphi(t)e^{n(t-\vartheta_2)}}{e^{n(t)}[\varphi(t)e^{\varphi(t)\vartheta_2}+c(t)(e^{\varphi(t)\vartheta_2}-1)e^{n(t-\vartheta_2)}]}dt\bigg] $.
\end{center}
Since $n$ is a periodic function in $X$, $n(t+\Theta)=n(t)$. Taking $t=0$ we get $n(\Theta)=n(0)$. Hence $ \int_{0}^{\Theta}n'(t)dt = n(\Theta)-n(0)=0$.
This implies
\begin{equation}\label{eq:a}
\begin{split}
\int_{0}^{\Theta}\dfrac{\varrho(t)\varphi(t)e^{n(t-\vartheta_1)}}{e^{n(t)}[\varphi(t)e^{\varphi(t)\vartheta_1}+c(t)(e^{\varphi(t)\vartheta_1}-1)e^{n(t-\vartheta_1)}]}dt=\int_{0}^{\Theta}\varphi(t)dt+\int_{0}^{\Theta}c(t)e^{n(t)}dt \\
+\int_{0}^{\Theta}\dfrac{q(t)E(t)\varphi(t)e^{n(t-\vartheta_2)}}{e^{n(t)}[\varphi(t)e^{\varphi(t)\vartheta_2}+c(t)(e^{\varphi(t)\vartheta_2}-1)e^{n(t-\vartheta_2)}]}dt.
\end{split}
\end{equation}
Hence
\begin{center}
$ \int_{0}^{\Theta}|n'(t)|dt < \Lambda\int_{0}^{\Theta}\dfrac{\varrho(t)\varphi(t)e^{n(t-\vartheta_1)}}{e^{n(t)}[\varphi(t)e^{\varphi(t)\vartheta_1}+c(t)(e^{\varphi(t)\vartheta_1}-1)e^{n(t-\vartheta_1)}]}dt+\Lambda\int_{0}^{\Theta}\varphi(t)dt$\\ 
$-\Lambda\int_{0}^{\Theta}c(t)e^{n(t)}dt $ $ -\Lambda\int_{0}^{\Theta}\dfrac{q(t)E(t)\varphi(t)e^{n(t-\vartheta_2)}}{e^{n(t)}[\varphi(t)e^{\varphi(t)\vartheta_2}+c(t)(e^{\varphi(t)\vartheta_2}-1)e^{n(t-\vartheta_2)}]}dt $.
\end{center}
Using equation \eqref{eq:a} we get,
\begin{center}
$ \int_{0}^{\Theta}|n'(t)|dt < 2\Lambda\int_{0}^{\Theta}\varphi(t)dt<2\int_{0}^{\Theta}\varphi(t)dt $ where $\Lambda\in (0,1) $.
\end{center}
We can say that
\begin{equation}\label{eq:b}
\int_{0}^{\Theta}|n'(t)|dt < 2\Tilde{\varphi}\Theta.
\end{equation}
We now denote
\begin{equation}\label{eq:d}
n(\xi_1)=\min\limits_{t\in[0,\Theta]}n(t) \, \textrm{and} \, n(\xi_2)=\max\limits_{t\in[0,\Theta]}n(t).
\end{equation}
From equations \eqref{eq:a}, \eqref{eq:d} we get
\begin{center}
$ \int_{0}^{\Theta}\dfrac{\varrho(t)\varphi(t)}{e^{n(\xi_1)}c(t)}dt \geqslant \int_{0}^{\Theta}c(t)e^{n(\xi_1)}dt $,
\end{center}
which is same as 
\begin{center}
$ \dfrac{1}{e^{n(\xi_1)}}\int_{0}^{\Theta}\dfrac{\varrho\varphi}{c}(t)dt\geqslant\Tilde{c}\Theta e^{n(\xi_1)} $.
\end{center}
That implies
\begin{center}
$ \dfrac{\Tilde{\varrho\varphi}}{\Tilde{c}}\geqslant \Tilde{c}e^{2n(\xi_1)} $.
\end{center}
Hence we obtain
\begin{equation}\label{eq:c}
n(\xi_1)\leqslant \dfrac{1}{2}\mbox{log}(\dfrac{\Tilde{\varrho\varphi}}{(\Tilde{c})^2})=j_1.
\end{equation} 
Adding equations \eqref{eq:b}, \eqref{eq:c} we get
\begin{center}
$ n(t)\leqslant n(\xi_1)+\int_{0}^{\Theta}|n '(t)|dt\leqslant \dfrac{1}{2}\mbox{log}(\dfrac{\Tilde{\varrho\varphi}}{(\Tilde{c})^2})+2\Tilde{\varphi}\Theta=R_1 $.
\end{center}
Hence we obtain
\begin{center}
$ n(t)\leqslant R_1 $.
\end{center}
From equations \eqref{eq:a},\eqref{eq:d} we get
\begin{center}
$\int_{0}^{\Theta} \varrho(t)\varphi(t) dt\leqslant\int_{0}^{\Theta}c(t)e^{n(\xi_2)}dt $,
\end{center}
which is same as
\begin{center}
$ \Tilde{\varrho\varphi} \leqslant e^{n(\xi_2)}\Tilde{c}$.
\end{center}
Using  $e^{n(\xi_2)}\Tilde{c} \leqslant e^{n(\xi_2)}\Tilde{c}+ (\Tilde{c}M^\varepsilon)^2 $, we get
\begin{center}
$ \Tilde{\varrho\varphi} \leqslant e^{n(\xi_2)}\Tilde{c}+ (\Tilde{c}M^\varepsilon)^2 $.
\end{center}
Hence we get
\begin{equation}\label{eq:e}
n(\xi_2)\geqslant \mbox{log}(\dfrac{\Tilde{\varrho\varphi}-(\Tilde{c}M^\varepsilon)^2}{\Tilde{c}})=j_2.
\end{equation}
Subtracting equation \eqref{eq:b} from equation \eqref{eq:e} we get
\begin{center}
$ n(t)\geqslant n(\xi_2)-\int_{0}^{\Theta}|n '(t)|dt\geqslant \mbox{log}(\dfrac{\Tilde{\varrho\varphi}-(\Tilde{c}M^\varepsilon)^2}{\Tilde{c}})-2\Tilde{\varphi}\Theta=R_2 $.
\end{center}
Hence we obtain
\begin{center}
$ n(t)\geqslant R_2 $.
\end{center}
Therefore $\max\limits_{t\in[0,\Theta]}|n(t)|\leqslant $ max $ \{|R_1|,|R_2|\}=D_1 $ ($ D_1 $ is independent of $ \Lambda $). \\
Take $ D=D_1+D_2 $, where $ D_2>0 $ such that $ D_2>|j_1|+|j_2| $ where $j_1=\dfrac{1}{2}\mbox{log}(\dfrac{\Tilde{\varrho\varphi}}{(\Tilde{c})^2})$ and $j_2=\mbox{log}(\dfrac{\Tilde{\varrho\varphi}-(\Tilde{c}M^\varepsilon)^2}{\Tilde{c}})$.
From the following equation where $ 0\leqslant\zeta\leqslant 1 $ is a parameter,
\begin{center}
$ \vspace{0.5cm}\int_{0}^{\Theta}\dfrac{\varrho(t)\varphi(t)e^{n(t-\vartheta_1)}}{e^{n}[\varphi(t)e^{\varphi(t)\vartheta_1}+c(t)(e^{\varphi(t)\vartheta_1}-1)e^{n(t-\vartheta_1)}]}dt-\int_{0}^{\Theta}\varphi(t)dt-\int_{0}^{\Theta}c(t)e^{n(t)}dt $ \vspace{-0.5cm}
\end{center}
\begin{center}
$ -\zeta \int_{0}^{\Theta}\dfrac{q(t)E(t)\varphi(t)e^{n(t-\vartheta_2)}}{e^{n}[\varphi(t)e^{\varphi(t)\vartheta_2}+c(t)(e^{\varphi(t)\vartheta_2}-1)e^{n(t-\vartheta_2)}]}dt=0 $,
\end{center}
It can be shown that a solution $ n^\ast $ to the above equation satisfies $ j_1\leqslant n^\ast\leqslant j_2. $ Let $\Omega=\{n\in X/ $ $ \|n \|< D\}.$ We can prove that $ \Omega $ obeys constraint (a) of Lemma \ref{lem:2}. \\
Suppose $ n\in $ Kernel $L$ $ \cap$ $\partial\Omega $ =$\mathbb{R}\cap\partial\Omega$, vector $ n$ is a constant, $ |n|=\|n\|=D $. From the definition of D and $ j_1\leqslant n^\ast\leqslant j_2, $ we have $ Sn=S_1(t) $.
Then, \vspace{0.3cm}
\begin{center}
$ QSn=QS_1(t)=\frac{1}{\Theta}\int_{0}^{\Theta}S_1(t)dt $.
\end{center}
That implies
\begin{center}
$\vspace{0.5cm} QSn=\dfrac{1}{\Theta}\int_{0}^{\Theta}\dfrac{\varrho(t)\varphi(t)e^{n(t-\vartheta_1)}}{e^{n}[\varphi(t)e^{\varphi(t)\vartheta_1}+c(t)(e^{\varphi(t)\vartheta_1}-1)e^{n(t-\vartheta_1)}]}-\Tilde{\varphi}-\Tilde{c}e^{n} $ \\ $ -\dfrac{1}{\Theta}\int_{0}^{\Theta}\dfrac{q(t)E(t)\varphi(t)e^{n(t-\vartheta_2)}}{e^{n}[\varphi(t)e^{\varphi(t)\vartheta_2}+c(t)(e^{\varphi(t)\vartheta_2}-1)e^{n(t-\vartheta_2)}]}dt\ne0 $. \vspace{0.3cm}
\end{center}
That is the first part of constraint (b) of Lemma \ref{lem:2}. \\
Let us have a look at the homotopy
\begin{center}
$ H_\zeta(n)=\zeta QS(n)+(1-\zeta)G(n), $ $ \zeta \in[0,1] $,
\end{center}
where
\begin{center}
$ G(n)= \dfrac{1}{\Theta}\int_{0}^{\Theta}\dfrac{\varrho(t)\varphi(t)e^{n(t-\vartheta_1)}}{e^{n}[\varphi(t)e^{\varphi(t)\vartheta_1}+c(t)(e^{\varphi(t)\vartheta_1}-1)e^{n(t-\vartheta_1)}]}-\Tilde{\varphi}-\Tilde{c}e^{n} $. \vspace{0.3cm}
\end{center}
From $ j_1\leqslant n^\ast\leqslant j_2, $ $ 0\notin H_\zeta $ (Kernel $ L$ $\cap$ $\partial\Omega$) for $ 0 \leqslant \zeta \leqslant 1. $ 
$ G(n)=0 $ possesses a unique solution that belongs to $ \mathbb{R}. $
As Image $ Q= $ Kernel $ L$ we have $ J=I $ and hence using homotopy invariance property, we get
$ deg\{JQS,\Omega $ $ \cap $ Kernel $ L,0\} $ is same as $deg\{QS,\Omega $ $ \cap $ Kernel $ L,0\}$ which is same as $deg\{G,\Omega $ $ \cap $ Kernel $ L,0\}$ which is not equal to zero.
Therefore using Lemma \ref{lem:2}, equation $ L\digamma=S\digamma $ has minimum one solution in the set $\Bar{\Omega} \cap $ Domain $ L $.
\, This implies, equation \eqref{eq:z} has minimum one $\Theta$-periodic solution in the set $\Bar{\Omega} \cap $ Domain $ L $, say $ n^\ast(t). $
Set $ \digamma^\ast(t)=exp\{n^\ast(t)\} $ and thus $ \digamma^\ast(t) $ becomes a $ \Theta$-periodic solution to system \eqref{eq:1}. Hence proved.
\end{proof}
\section{Existence of almost periodic solution}\label{s7}
It is very well known that the prey-predator relationships in day to day life undergo several types of perturbations, among which certain perturbations are periodic \cite{FK04}. Sometimes the perturbations can be quasi-periodic or almost periodic.
A function $f:\mathbb{R}\rightarrow\mathbb{R}$ is $\ni$ $f(t+T)=f(t)$ holds for infinitely many values of $T$, with a random degree of precision ($T$ is spread over $-\infty$ to $+\infty$) in a way that empty intervals of arbitrarily great distance are not left, then $f$ is said to be almost periodic \cite{TJP16}. Now let us assume that $E(t)$, $\varphi(t)$, $c(t)$, $\varrho(t)$ and $q(t)$ are functions in $t$(time) that satisfy almost periodicity property.
\begin{lemma}\label{lem:3}
(Arzela-Ascoli theorem)\normalfont{\cite{BB64}} \textit{Suppose that $f$, $g$ are any two positive integers, $K$ $\subset$ $\mathbb{R}^f$ and $K$ is compact in it and $\sigma=\{j/ j\in C(K,\mathbb{R}^g)\}$, then below mentioned properties are identical: \\
(a) $\sigma$ is a bounded set and is equi-continuous on $K$. \\
(b) Each sequence in $\sigma$ contains a subsequence that becomes convergent in $K$, uniformly.}
\end{lemma}
\begin{theorem}\label{thm:t6}
If every constraint in Theorem \ref{thm:t3} is valid, there is a unique solution to system \eqref{eq:1} which is almost periodic in nature.
\end{theorem}
\begin{proof}
We know that all the solutions to system \eqref{eq:1}, satisfying the criteria given in Theorem \ref{thm:t2}, will be ultimately bounded above. Therefore, it is said that $\exists$ a bounded (bordered by positive values) positive($>0$) solution $r(t)$ to system \eqref{eq:1}. Thus $\exists$ a sequence $\{t_n\}$, $t_n \rightarrow\infty$ as $n\rightarrow\infty$ $\ni$ $r(t+t_n)$ satisfies the next system,
\begin{center}
$\vspace{0.3cm}\digamma^{'}(t)=\dfrac{\varrho(t+t_n)\varphi(t+t_n)\digamma(t-\vartheta_1)}{\varphi(t+t_n)e^{\varphi(t+t_n)\vartheta_1}+c(t+t_n)(e^{\varphi(t+t_n)\vartheta_1}-1)\digamma(t-\vartheta_1)}-\varphi(t+t_n)\digamma(t)-c(t+t_n)\digamma^2(t)$\\
$-\dfrac{q(t+t_n)E(t+t_n)\varphi(t+t_n)\digamma(t-\vartheta_2)}{\varphi(t+t_n)e^{\varphi(t+t_n)\vartheta_2}+c(t+t_n)(e^{\varphi(t+t_n)\vartheta_2}-1)\digamma(t-\vartheta_2)}\cdot$
\end{center}
Thus ${\Dot{r}(t+t_n)}$ (derivative of ${r(t+t_n)}$) and ${r(t+t_n)}$ will be bounded uniformly and will also be equi-continuous. 
By Lemma \ref{lem:3}, for a sequence $\{r(t+t_n)\}$ $\exists$ a subsequence $\{r(t+t_m)\}$ that is uniformly convergent. 
And for any $\varepsilon>0$ $\exists$ $\rho(\varepsilon)>0$ $\ni$ $|r(t+t_k)-r(t+t_m)|<\varepsilon$ if $k,m>\rho(\varepsilon)>0$. Thus it can be deduced that $r(t)$ is asymptotically almost periodic. Thus $\{r(t+t_m)\}$ is expressed as $r(t+t_m)=r_1(t+t_m)+r_2(t+t_m)$ where function $r_1(t+t_m)$ is almost periodic and function $r_2(t+t_m)$ is continuous $\forall$ $t\in\mathbb{R}$.
\par We also have $\lim\limits_{m\rightarrow\infty}r_2(t+t_m)=0$, $\lim\limits_{m\rightarrow\infty}r_1(t+t_m)=r_1(t)$ where function $r_1$ is almost periodic. Hence $\lim\limits_{m\rightarrow\infty}r(t+t_m)=r_1(t)$.\\
Also, 
\begin{align*}
\vspace{0.3cm}\lim\limits_{m\rightarrow\infty}\Dot{r}(t+t_m)=\lim\limits_{m\rightarrow\infty}\lim\limits_{h\rightarrow0}\dfrac{r(t+t_m+h)-r(t+t_m)}{h}\vspace{0.5cm}\\
=\lim\limits_{h\rightarrow0}\lim\limits_{m\rightarrow\infty}\dfrac{r(t+t_m+h)-r(t+t_m)}{h}
\end{align*}
\begin{align*}
=\lim\limits_{h\rightarrow0}\dfrac{r_1(t+h)-r_1(t)}{h}\cdot 
\end{align*}
Thus, we can say that $\Dot{r_1}$ exists.
There is a sequence $\{t_n\}$ $\ni$ $t_n\rightarrow\infty$ as $n\rightarrow\infty$ in which $\varrho(t+t_n)\rightarrow \varrho(t)$, $c(t+t_n)\rightarrow c(t)$, $\varphi(t+t_n)\rightarrow\varphi(t)$, 
$E(t+t_n)\rightarrow E(t)$ and  $q(t+t_n)\rightarrow q(t)$. 
\begin{align*}
\vspace{1cm}\Dot{r_1}=\lim\limits_{n\rightarrow\infty}\dfrac{d}{dt}r(t+t_n)=\lim\limits_{n\rightarrow\infty}\bigg[\dfrac{\varrho(t+t_n)\varphi(t+t_n)r(t+t_n-\vartheta_1)}{\varphi(t+t_n)e^{\varphi(t+t_n)\vartheta_1}+c(t+t_n)(e^{\varphi(t+t_n)\vartheta_1}-1)r(t+t_n-\vartheta_1)}\vspace{0.5cm}\\ -\varphi(t+t_n)r(t+t_n)-c(t+t_n)r^2(t+t_n)\\ \vspace{0.5cm}-\dfrac{q(t+t_n)E(t+t_n)\varphi(t+t_n)r(t+t_n-\vartheta_2)}{\varphi(t+t_n)e^{\varphi(t+t_n)\vartheta_2}+c(t+t_n)(e^{\varphi(t+t_n)\vartheta_2}-1)r(t+t_n-\vartheta_2)}\bigg]\cdot \vspace{0.1cm}
\end{align*}
We get,
\begin{flalign*}
\Dot{r_1}=\dfrac{\varrho(t)\varphi(t)r(t-\vartheta_1)}{\varphi(t)e^{\varphi(t)\vartheta_1}+c(t)(e^{\varphi(t)\vartheta_1}-1)r(t-\vartheta_1)}-\varphi(t)r(t)-c(t)r^2(t)\\
-\dfrac{q(t)E(t)\varphi(t)r(t-\vartheta_2)}{\varphi(t)e^{\varphi(t)\vartheta_2}+c(t)(e^{\varphi(t)\vartheta_2}-1)r(t-\vartheta_2)}\cdot
\end{flalign*}
Thus we can conclude that $r_1$ satisfies system \eqref{eq:1} and it is almost periodic. Hence, we proved that system \eqref{eq:1} possesses a unique solution that is both positive and almost periodic in nature.
\end{proof}
\textbf{Remark}: Results proved in Theorems \ref{thm:t1},\ref{thm:t2} and \ref{thm:t3} for positive invariance, permanence and global attractivity remain valid for system \eqref{eq:1} with almost periodic coefficients. 
\section{Numerical simulation}\label{s8}
\textbf{Example 8.1.}
Let $\varrho(t)=1.3$, $c(t)=0.005$, $\varphi(t)=0.02+0.01(\cos t)$, $q(t)=0.9$,
$E(t)=1.3$, $\vartheta_1=2$, $\vartheta_2=1$, then system \eqref{eq:1} becomes
\begin{center}\vspace{-0.1cm}
$\digamma^{'}(t)=\dfrac{(1.3)(0.02+0.01(\cos t))\digamma(t-2)}{(0.02+0.01(\cos t))e^{2(0.02+0.01(\cos t))}+(0.005)(e^{2(0.02+0.01(\cos t))}-1)\digamma(t-2)}$
\end{center}
\begin{center}\vspace{-0.1cm}
$-(0.02+0.01(\cos t))\digamma(t) - (0.005)\digamma^2(t)$ 
\end{center}
\begin{center}\vspace{-0.1cm}
$-\dfrac{(0.9)(1.3)(0.02+0.01(\cos t))\digamma(t-1)}{(0.02+0.01(\cos t))e^{(0.02+0.01(\cos t)}+(0.005)(e^{0.02+0.01(\cos t)}-1)\digamma(t-1)}\cdot$
\end{center}
By simple numerical computation, \vspace{0.1cm}
$ M^\varepsilon =\dfrac{\varrho_M\varphi_M}{c_L}+\varepsilon=7.8001$, $ m^\varepsilon=\dfrac{\varrho_L\varphi_L-(c_M)^2(M^\varepsilon)^2 }{c_M} -\varepsilon=2.2959 $ for $\varepsilon=0.0001$. 
$ K_\varepsilon=\{\digamma(t)\in R / m^\varepsilon \leqslant \digamma(t) \leqslant M^\varepsilon \} $ = $\{\digamma(t)\in R / 2.2959 \leqslant \digamma(t) \leqslant 7.8001 \}$ and 
$\varrho_L\varphi_L=0.013$, $(c_M)^2(M^\varepsilon)^2=0.001521$ which implies that $\varrho_L\varphi_L > (c_M)^2(M^\varepsilon)^2$. \\ 
Then by Theorem \ref{thm:t1} we can say that $ K_\varepsilon$ satisfies the property of positive invariance (w.r.t system \eqref{eq:1}), which can also be seen from Figure \ref{fig:1}. \\
By taking $\varepsilon=0$ in Theorem \ref{thm:t1} we obtain Theorem \ref{thm:t2} in which $M^\varepsilon=M^0=7.8$ and  $m^\varepsilon=m^0=2.296$. From Figure \ref{fig:1} it is clear that system \eqref{eq:1} is permanent.\\
\textbf{Example 8.2.}
Let Let $\varrho(t)=1.3$, $c(t)=0.029+0.05(\sin^2t)$, $\varphi(t)=0.033+0.01(\sin t)$, $q(t)=0.88$,
$E(t)=1.2$, $\vartheta_1=2$, $\vartheta_2=1$, then system \eqref{eq:1} becomes \vspace{0.05cm}\\
$\digamma^{'}(t)=$ \vspace{-0.05cm}
\begin{center}
$\dfrac{(1.3)(0.033+0.01(\sin t))\digamma(t-2)}{(0.033+0.01(\sin t))e^{2(0.033+0.01(\sin t))}+(0.029+0.05(\sin^2(t)))(e^{2(0.033+0.01(\sin t))}-1)\digamma(t-2)}$
\end{center}
\begin{center}
$-(0.033+0.01(\sin t))\digamma(t) - (0.029+0.05(\sin^2t))\digamma^2(t)$
\end{center}
\begin{center}\vspace{-0.05cm}
$-\dfrac{(0.88)(1.2)(0.033+0.01(\sin t))\digamma(t-1)}{(0.033+0.01(\sin  t))e^{(0.033+0.01(\sin t)}+(0.029+0.05(\sin^2(t)))(e^{0.033+0.01(\sin t)}-1)\digamma(t-1)}\cdot$ 
\end{center}\vspace{-0.2cm}
We can calculate 
$ M^\varepsilon =\dfrac{\varrho_M\varphi_M}{c_L}+\varepsilon=1.94247931$, $ m^\varepsilon=\dfrac{\varrho_L\varphi_L-(c_M)^2(M^\varepsilon)^2 }{c_M} -\varepsilon=0.3602627 $ for $\varepsilon=0.0001$. 
And $\varrho_L\varphi_L=0.030329$, $(c_M)^2(M^\varepsilon)^2=0.0235487$ which implies that $\varrho_L\varphi_L > (c_M)^2(M^\varepsilon)^2$. 
By taking $\varepsilon=0$ we get $M^0=1.94237931$, $m^0=0.3603627$. Also, $\varrho_L\varphi_L>(c_M)^2(M^0)^2$.
From Theorem \ref{thm:t2} we can say that system \eqref{eq:1} is permanent, which is shown in Figure \ref{fig:2}. Also, $ K_\varepsilon=\{\digamma(t)\in R / m^\varepsilon \leqslant \digamma(t) \leqslant M^\varepsilon \} $ = $\{\digamma(t)\in R / 0.3602627 \leqslant \digamma(t) \leqslant 1.94247931 \}$ and we just showed that $\varrho_L\varphi_L > (c_M)^2(M^\varepsilon)^2$. Therefore, from Theorem \ref{thm:t1}, $ K_\varepsilon$ satisfies positive invariance property (w.r.t system \eqref{eq:1}), which can also be seen in Figure \ref{fig:2}.\\
\textbf{Example 8.3.}
Let $\varrho(t)=4.9+0.11(\cos^2t)$, $c(t)=6.4-0.025(\sin t)$, $\varphi(t)=0.02+0.01\cos t$, 
$q(t)=3.29$, $E(t)=0.9$, $\vartheta_1=2$, $\vartheta_2=1$, then system \eqref{eq:1} becomes \vspace{0.05cm}\\ 
$\digamma^{'}(t)=$
\begin{center}\vspace{-0.2cm}
$\dfrac{(4.9+0.11(\cos^2t))(0.02+0.01(\cos t))\digamma(t-2)}{(0.02+0.01(\cos t))e^{2(0.02+0.01(\cos t))}+(6.4-0.025(\sin t))(e^{2(0.02+0.01(\cos t))}-1)\digamma(t-2)}$
\end{center}
\begin{center}
%\vspace{-0.09cm}%
$-(0.02+0.01(\cos t))\digamma(t) - (6.4-0.025(\sin t))\digamma^2(t)$
\end{center}
\begin{center}
%\vspace{-0.09cm}%
$-\dfrac{(3.29)(0.9)(0.02+0.01(\cos t))\digamma(t-1)}{(0.02+0.01(\cos t))e^{(0.02+0.01(\cos t)}+(6.4-0.025(\sin t))(e^{0.02+0.01(\cos t)}-1)\digamma(t-1)}\cdot$
\end{center}\vspace{-0.2cm}
For $\varepsilon=0.0001$, we can compute 
%$ M^\varepsilon =\dfrac{\varrho_M\varphi_M}{c_L}+\varepsilon=0.023682$, $ m^\varepsilon=\dfrac{\varrho_L\varphi_L-(c_M)^2(M^\varepsilon)^2 }{c_M} -\varepsilon=0.0039534 $ for $\varepsilon=0.0001$. \vspace{0.05cm}\\
$\varrho_L\varphi_L=0.049$, $(c_M)^2(M^\varepsilon)^2=0.0231517$ which implies that $\varrho_L\varphi_L > (c_M)^2(M^\varepsilon)^2$. 
Also from equation \eqref{eq:2},
\begin{center}
$\dfrac{-\varrho(t)\varphi^2(t)e^{\varphi(t)\vartheta_1}}{(\varphi(t)e^{\varphi(t)\vartheta_1}+c(t)(e^{\varphi(t)\vartheta_1}-1)M^\varepsilon)^2} + 4c(t)m^\varepsilon + 2\varphi(t)
+ \dfrac{q(t)E(t)\varphi^2(t)e^{\varphi(t)\vartheta_2}}{(\varphi(t)e^{\varphi(t)\vartheta_2}+c(t)(e^{\varphi(t)\vartheta_2}-1)M^\varepsilon)^2}$
%\vspace{-0.15cm}%
\end{center}
$=0.063903198 > 0$. Hence from Theorem \ref{thm:t3}, $ \digamma^\ast(t) $, a bounded (enclosed within positive values) and positive ($>0$) solution to system \eqref{eq:1} is globally attractive, as seen in Figure \ref{fig:3}(a).
The difference between the two trajectories starting at initial values $\digamma(t_0)=0.06$ and $\digamma(t_0)=0.07$ is plotted against time $t$ in Figure \ref{fig:3}(b). It is now clear that as $t \to \infty $ the difference between the two trajectories tends to zero.\\
\textbf{Example 8.4.}
Let $\varrho(t)=8-(\sin t)$, $c(t)=7+(\cos t)$, $\varphi=0.023$,
$q(t)=0.48$, $E(t)=9$, $\vartheta_1=2$, $\vartheta_2=1$, then system \eqref{eq:1} becomes
\begin{center}
$\digamma^{'}(t)=\dfrac{(8-(\sin t))(0.023)\digamma(t-2)}{(0.023)e^{2(0.023)}+(7+(\cos t))(e^{2(0.023)}-1)\digamma(t-2)}$\vspace{-0.2cm}
\end{center}
\begin{center}
$-(0.023)\digamma(t) -(7+\cos t))\digamma^2(t)$
$-\dfrac{(0.48)(9)(0.023)\digamma(t-1)}{(0.023)e^{0.023}+(7+(\cos t))(e^{0.023}-1)\digamma(t-1)}\cdot$\vspace{-0.2cm}
\end{center}
By simple numerical computation, we have 
%$ M^\varepsilon =\dfrac{\varrho_M\varphi_M}{c_L}+\varepsilon=0.0346$, $ m^\varepsilon=\dfrac{\varrho_L\varphi_L-(c_M)^2(M^\varepsilon)^2 }{c_M} -\varepsilon=0.010503 $ for $\varepsilon=0.0001$. \vspace{0.2cm} \\%
$\varrho_L\varphi_L=0.161$, $(c_M)^2(M^\varepsilon)^2=0.07661824$ for $\varepsilon=0.0001$. This implies that $\varrho_L\varphi_L > (c_M)^2(M^\varepsilon)^2$. 
Also from equation \eqref{eq:2},
\begin{center}\vspace{-0.2cm}
$\dfrac{-\varrho(t)\varphi^2(t)e^{\varphi(t)\vartheta_1}}{(\varphi(t)e^{\varphi(t)\vartheta_1}+c(t)(e^{\varphi(t)\vartheta_1}-1)M^\varepsilon)^2} + 4c(t)m^\varepsilon + 2\varphi(t) 
+ \dfrac{q(t)E(t)\varphi^2(t)e^{\varphi(t)\vartheta_2}}{(\varphi(t)e^{\varphi(t)\vartheta_2}+c(t)(e^{\varphi(t)\vartheta_2}-1)M^\varepsilon)^2}$\vspace{-0.2cm}
\end{center}
$=0.08561413 > 0$. Hence from Theorem \ref{thm:t3}, $ \digamma^\ast(t) $, a bounded (enclosed within positive values) and positive ($>0$) solution to system \eqref{eq:1} is globally attractive, as shown in Figure \ref{fig:3}(c). 
The difference between the two trajectories starting at initial values $\digamma(t_0)=0.1$ and $\digamma(t_0)=0.2$ is plotted against time $t$ in Figure \ref{fig:3}(d). It is now clear that as $t \to \infty $ the difference between the two trajectories tends to zero. \vspace{0.1cm}\\
\textbf{Example 8.5.}
Let $\varrho(t)=1.3+0.33\sin t$,  $c(t)=0.2-0.1428\cos t$,
$\varphi=0.02+0.01\cos t$, 
$q(t)=0.025-0.0166\sin t$, $E(t)=1-0.5\cos^2t$, $\vartheta_1=2$, $\vartheta_2=1$ and choose $\Theta=2\pi$ then system \eqref{eq:1} becomes \\
$\digamma^{'}(t)=$
\begin{center}\vspace{-0.2cm}
$\dfrac{(1.3+0.33\sin t)(0.02+0.01\cos t)\digamma(t-2)}{(0.02+0.01\cos t)e^{2(0.02+0.01\cos t)}+(0.2-0.1428\cos t)(e^{2(0.02+0.01\cos t)}-1)\digamma(t-2)}$
\end{center}
\begin{center}\vspace{-0.2cm}
$-(0.02+0.01(\cos t))\digamma(t) -(0.2-0.1428\cos t))\digamma^2(t)$
\end{center}
\begin{center}\vspace{-0.2cm}
$-\dfrac{(0.025-0.0166\sin t)(1-0.5\cos^2t))(0.02+0.01\cos t)\digamma(t-1)}{(0.02+0.01\cos t)e^{0.02+0.01\cos t}+(0.2-0.1428\cos t)(e^{0.02+0.01\cos t}-1)\digamma(t-1)}\cdot$
\vspace{0.05cm}\end{center}
For $\varepsilon=0.0001$, we can compute 
%$ M^\varepsilon =\dfrac{\varrho_M\varphi_M}{c_L}+\varepsilon=0.14301381$ for $\varepsilon=0.0001$. \vspace{0.1cm} \\%
$ \Tilde{\varrho\varphi}=0.013 $, $(\Tilde{c} M^\varepsilon)^2=0.00081812 $
which implies that $\Tilde{\varrho\varphi}>(\Tilde{c} M^\varepsilon)^2. $ 
Hence from Theorem \ref{thm:t5}, system \eqref{eq:1} possesses not less than one positive($>0$) $ \Theta $-periodic solution, which is also shown in Figure \ref{fig:4}(a).\\
\textbf{Example 8.6.}
Let $\varrho(t)=7+0.33\cos t$,  $c(t)=5-0.25\sin t$, 
$\varphi=0.02+0.01\cos t$, 
$q(t)=0.025-0.0166\sin t$, $E(t)=1-0.5\cos^2t$, $\vartheta_1=2$, $\vartheta_2=1$ and choose $\Theta=2\pi$ then system \eqref{eq:1} becomes \vspace{0.2cm}\\
$\digamma^{'}(t)=$
\begin{center}\vspace{-0.2cm}
$\dfrac{(7+0.33\cos t)(0.02+0.01\cos t)\digamma(t-2)}{(0.02+0.01\cos t)e^{2(0.02+0.01\cos t)}+(5-0.25\sin t)(e^{2(0.02+0.01\cos t)})-1)\digamma(t-2)}$
\end{center}
\begin{center}
$-(0.02+0.01\cos t)\digamma(t) -(5-0.25\sin t)\digamma^2(t)$
\end{center}
\begin{center}
$-\dfrac{(0.025-0.0166\sin t)(1-0.5\cos^2t))(0.02+0.01\cos t))\digamma(t-1)}{(0.02+0.01\cos t)e^{0.02+0.01\cos t}+(5-0.25\sin t)(e^{0.02+0.01\cos t}-1)\digamma(t-1)}\cdot$
\end{center}
By simple numerical computation,
%$ M^\varepsilon =\dfrac{\varrho_M\varphi_M}{c_L}+\varepsilon=0.04537486$ for $\varepsilon=0.0001$. \vspace{0.1cm}\\%
$ \Tilde{\varrho\varphi}=0.07 $, $(\Tilde{c} M^\varepsilon)^2=0.0514719$ for $\varepsilon=0.0001$.
This implies that $\Tilde{\varrho\varphi}>(\Tilde{c} M^\varepsilon)^2. $ 
Hence we have, from Theorem \ref{thm:t5}, system \eqref{eq:1} possesses minimum one positive($>0$) $ \Theta $-periodic solution, as shown in Figure \ref{fig:4}(b). \vspace{0.2cm}\\
\textbf{Example 8.7.}
Let $\varrho(t)=14+0.1\cos^2t$,  $c(t)=2+0.33\sin t$,
$\varphi=0.032$, 
$q(t)=0.03-0.01\sin t$, $E(t)=235$, $\vartheta_1=2$, $\vartheta_2=1$ and choose $\Theta=2\pi$ then system \eqref{eq:1} becomes
\begin{center}
$\digamma^{'}(t)=$ $\dfrac{(14+0.1\cos^2t)(0.032)\digamma(t-2)}{(0.032)e^{2(0.032)}+(2+0.33\sin t)(e^{2(0.032)})-1)\digamma(t-2)}$
\end{center}
\begin{center}\vspace{-0.2cm}
$-(0.032)\digamma(t) -(2+0.33\sin t)\digamma^2(t)$
$-\dfrac{(0.03-0.01\sin t)(235)(0.032)\digamma(t-1)}{(0.032)e^{0.032}+(2+0.33\sin t)(e^{0.032}-1)\digamma(t-1)}\cdot$ \vspace{-0.3cm}
\end{center}
We can calculate 
%$ M^\varepsilon =\dfrac{\varrho_M\varphi_M}{c_L}+\varepsilon=0.2699679$ for $\varepsilon=0.0001$. \vspace{0.1cm}\\%
$ \Tilde{\varrho\varphi}=0.4496 $, $(\Tilde{c} M^\varepsilon)^2=0.2915307$ for $\varepsilon=0.0001$.
This implies that $\Tilde{\varrho\varphi}>(\Tilde{c} M^\varepsilon)^2. $ 
Hence from Theorem \ref{thm:t5}, system \eqref{eq:1} has one or more positive($>0$) $ \Theta $-periodic solutions. From Figure \ref{fig:5} we can observe that system \eqref{eq:1} admits triple-periodic solution.
Also, a double periodic solution can be seen in Figure \ref{fig:2}. \vspace{0.2cm}\\
\textbf{Example 8.8.} \vspace{0.2cm}
Let $\varrho(t)=6+\dfrac{\cos(t)}{8}$, $c(t)=0.01+0.0011\cos t$, $\varphi(t)=0.1$, $q(t)=0.3$,\\
$E(t)=10$, $\vartheta_1=2$, $\vartheta_2=1$ and choose $\Theta=2\pi$ then system \eqref{eq:1} becomes \vspace{0.1cm}\\
$\digamma^{'}(t)=$ \vspace{-0.6cm}
\begin{center}
$\dfrac{(6+\dfrac{\cos t}{8})(0.1)\digamma(t-2)}{(0.1)e^{2(0.1)}+(0.01+0.0011\cos t)(e^{2(0.1)}-1)\digamma(t-2)}$
$-(0.1)\digamma(t)$
\end{center}
\begin{center}\vspace{-0.2cm}
$-(0.01+0.0011\cos t)\digamma^2(t)$
$-\dfrac{(0.3)(10)(0.1)\digamma(t-1)}{(0.1)e^{0.1}+(0.01+0.0011\cos t)(e^{0.1}-1)\digamma(t-1)}\cdot$ 
\end{center}\vspace{-0.2cm}
We can calculate 
\vspace{0.1cm} $ M^\varepsilon =\dfrac{\varrho_M\varphi_M}{c_L}+\varepsilon=68.8977378$, $ m^\varepsilon=\dfrac{\varrho_L\varphi_L-(c_M)^2(M^\varepsilon)^2 }{c_M} -\varepsilon=0.1423013 $ for $\varepsilon=0.0001$. 
And $\varrho_L\varphi_L=0.5875$, $(c_M)^2(M^\varepsilon)^2=0.584865337$ which implies that $\varrho_L\varphi_L > (c_M)^2(M^\varepsilon)^2$. 
By taking $\varepsilon=0$, we get $M^0=68.8976378$, $m^0=0.1424013$. Also, $\varrho_L\varphi_L>(c_M)^2(M^0)^2$.
From Theorem \ref{thm:t2}, system \eqref{eq:1} is permanent, which is shown in Figure \ref{fig:6}.\\ Also, $ K_\varepsilon=\{\digamma(t)\in R / m^\varepsilon \leqslant \digamma(t) \leqslant M^\varepsilon \} $ = $\{\digamma(t)\in R / 0.1423013\leqslant \digamma(t) \leqslant 68.8977378 \}$ \\ and we just showed that $\varrho_L\varphi_L > (c_M)^2(M^\varepsilon)^2$. Therefore, from Theorem \ref{thm:t1}, $ K_\varepsilon$ satisfies the property of positive invariance (w.r.t system \eqref{eq:1}), which can also be seen in Figure \ref{fig:6}.\\
Also from equation \eqref{eq:2},\vspace{0.2cm}
\begin{center}
$\dfrac{-\varrho(t)\varphi^2(t)e^{\varphi(t)\vartheta_1}}{(\varphi(t)e^{\varphi(t)\vartheta_1}+c(t)(e^{\varphi(t)\vartheta_1}-1)M^\varepsilon)^2} + 4c(t)m^\varepsilon + 2\varphi(t) 
+ \dfrac{q(t)E(t)\varphi^2(t)e^{\varphi(t)\vartheta_2}}{(\varphi(t)e^{\varphi(t)\vartheta_2}+c(t)(e^{\varphi(t)\vartheta_2}-1)M^\varepsilon)^2}$
\vspace{0.1cm}
\end{center}
$=0.269876035 > 0$. Hence from Theorem \ref{thm:t3}, $ \digamma^\ast(t) $, bounded (bordered by positive values) positive ($>0$) solution to system \eqref{eq:1} is globally attractive, as seen in Figure \ref{fig:7}(a). The difference between two trajectories starting with initial values $\digamma(t_0)=40$ and $\digamma(t_0)=50$ is plotted over time $t$ in Figure \ref{fig:7}(b). It is clear that as $t \to \infty $ the difference between these trajectories tends to zero. 
\par And $ \Tilde{\varrho\varphi}=0.6 $, $(\Tilde{c} M^\varepsilon)^2 =0.474689827 $ which implies that $\Tilde{\varrho\varphi}>(\Tilde{c} M^\varepsilon)^2$. Hence using Theorem \ref{thm:t5}, there is minimum one positive($>0$) $\Theta$-periodic solution to system \eqref{eq:1} as shown in Figure \ref{fig:8}. A double periodic solution can be noticed in Figure \ref{fig:9}. 
\par On the other hand, let $\varrho(t)=6+\frac{\cos((10\pi/\sqrt3)t)}{8}$, i.e, period is rationally independent. According to theorem \ref{thm:t6}, there is an almost periodic solution as shown in Figure \ref{fig:10}. \vspace{0.3cm}\\
\textbf{Example 8.9.} Let $\varrho(t) = 6+\frac{cos(t)}{8}$,
$\varphi(t) = 0.10$, $c(t) = 0.01 - (0.0011)cos(t)$,
$q(t) = 0.3$, $E(t) = 10$. We plot the population curves for different values of $\vartheta_1$ and $\vartheta_2$ in figures \ref{fig:11} and \ref{fig:12}. Our observations are shown in Tables \ref{Table:1} and \ref{Table:2}.\vspace{0.3cm}\\
%Population with initial value 30','Population with initial value 40%
\textbf{Example 8.10.} Let $\varrho(t) = 1.3$, $\varphi(t) = 0.02 + (0.01)cos(t)$, $c(t) = 0.005$, $q(t) = 0.9$, $E(t) = 1.3$. The population curves for different values of $\vartheta_1$ and $\vartheta_2$ are shown in figures \ref{fig:13} and \ref{fig:14}. Our observations are shown in Tables \ref{Table:3} and \ref{Table:4}.
%'Population with initial value 3','Population with initial value 4'%
\section{Conclusion}\label{s9}
The dynamical characteristics of any real-world ecological model depend explicitly on time. Most of the environmental factors vary with time and thus incorporating temporal inhomogeneity into an ecological model helps to understand it better. \\

The present work is essentially concerned with the age-based growth model of hilsa fish with time-variant parameters. We have obtained the sufficient condition for positive invariance from Theorem \ref{thm:t1} and the sufficient condition for permanence from Theorem \ref{thm:t2}, respectively. Also, if $\varepsilon=0$ in Theorem \ref{thm:t1}, we obtain the condition in Theorem \ref{thm:t2}. This means that if $K_\varepsilon$ is positively invariant in system \eqref{eq:1}, then system \eqref{eq:1} must be permanent, which is pretty clear from figures \ref{fig:1}, \ref{fig:2} and \ref{fig:6}. That is, if the population of hilsa lies within a particular range (given by $ [m^\varepsilon, M^\varepsilon]$ (see Theorem \ref{thm:t1})) and lasts forever (as $t \rightarrow \infty$) in that particular range then the population would never go extinct. In addition, we attain suitable conditions for global attractivity of a bounded positive solution by formulating a scalar function, namely the Lyapunov function in Theorem \ref{thm:t3}. In figures\ref{fig:3}(a) and (c) and \ref{fig:7}(a), there is one such trajectory (hilsa population) in each of these figures that attracts the other trajectory (hilsa population starting at a different initial value) towards it, proving the global stability of hilsa population. \\

It is known that birth, growth, spawning and, immigration of hilsa is very much seasonal, thus indicating that the environmental factors that influence the life cycle of hilsa are seasonal. Taking this into account, in Theorems \ref{thm:t4} and \ref{thm:t5}, with the help of Brouwer fixed point and continuation theorems we have obtained some constraints for a positive periodic solution of system \eqref{eq:1}. We can see that if $\varepsilon$ is chosen suitably and condition in Theorem \ref{thm:t1} holds, then the condition in Theorem \ref{thm:t5} holds, that is, Theorem \ref{thm:t5} is weaker than Theorem \ref{thm:t1}. Finally, in Theorem \ref{thm:t6} we have proved that there is only one solution to system \eqref{eq:1} which is almost periodic in nature that is distinct from every other solution. Figures \ref{fig:4}, \ref{fig:5}, \ref{fig:8} and \ref{fig:9} show that hilsa population follows a periodic cycle and figure \ref{fig:10} shows that hilsa population can also be quasi-periodic. This means that the fluctuations in the hilsa population follow a regular repeating pattern. For example, in Figure \ref{fig:4} we see that the changes in hilsa population that includes changes due to various factors like birth, death, food resources, fishing, etc, repeat every five years. This before hand prediction helps in taking precautionary measures to prevent a fall in the hilsa population.  \\

Hilsa stocks have been fast decreasing in the river systems. Thus conserving the population of hilsa and maintaining its global stability is one of the major concerns of ecologists across India and Bangladesh. The role played by $\vartheta_1$ and $\vartheta_2$ (delay in maturity and delay in harvest respectively) in shaping the dynamical aspects of the system \eqref{eq:1} of hilsa is as shown in Tables \ref{Table:1}, \ref{Table:2}, \ref{Table:3} and \ref{Table:4}. In Table \ref{Table:1}, we see that when the delay in maturity $\vartheta_1$ increases to 2.208 (critical value of $\vartheta_1$) with $\vartheta_2$ remaining constant at 1, there is a drastic change in the dynamical characteristics of the system. Then if $\vartheta_2$ increases beyond 1, the system regains its dynamical characteristics. This means that for a maturity delay of 2.208, the delay in harvest must be greater than 1. Similarly, from Table \ref{Table:3}, for a maturity delay of 2.24, the delay in harvest must be greater than 1 (see Figures \ref{fig:11}, \ref{fig:12}, \ref{fig:13} and \ref{fig:14}). This interplay between $\vartheta_1$ and $\vartheta_2$ plays a significant role in shaping the dynamical characteristics of the system \eqref{eq:1}. Therefore, a change in the fishing pattern, by paying extra attention to immature hilsa and taking into account the size and body weight of hilsa before harvesting, and a change in the approach towards the ecological model of hilsa, by acknowledging the importance of time-variant parameters, can stop hilsa from becoming extinct.

\section*{Appendix A}

The rate at which the population changes depends on the following aspects: birth rate, rate of death and intraspecific competitiveness (either crowding or head-on intervention). Death rate and intraspecific competition rate both together contribute to the decline rate, which is instantaneous. Considering this we revise the standard logistic ordinary differential equation \cite{AWW06} into:
\[\frac{d\digamma(t)}{dt}=\varrho\digamma(t)-\varphi \digamma(t)-c\digamma^2(t),\eqno{(A1)}\]
where $\varrho$, $\varphi$ and $c$ denote the birth, death and removal (due to intraspecific competition) coefficients respectively for $\digamma(t)$. If $r=\varrho-\varphi$ (i.e. intrinsic growth rate of $\digamma(t)$) and $K=\frac{\varrho-\varphi}{c}$ (i.e. environmental carrying capacity of $\digamma(t)$), then Eq. $(A1)$ becomes the classical logistic equation:
\[\frac{d\digamma(t)}{dt}=r\digamma(t)\bigg(1-\frac{\digamma(t)}{K}\bigg).\eqno{(A2)}\]
In $(A1)$, the rate of decline and rate of growth (both of which depend on the density of population at the current instant $t$) are denoted by $\{\varphi \digamma(t)+c\digamma^2(t)\}$ and $\varrho\digamma(t)$ respectively. 
In order to acquire an equation to describe the number of
individuals that are alive at $t-\vartheta$ shall be still alive at $t$, consider a first ordered ODE of $\digamma(t)$ as a function
of $\digamma(t-\vartheta)$ and solve it as follows:
\[\begin{array}{l}
\dfrac{d\digamma(t)}{dt}=-\varphi \digamma(t)-c\digamma^2(t).
\end{array}
\]
Using variable separable technique and then
integrating from $t-\vartheta$ to $t$, we get
\[\begin{array}{l}
\int_{\digamma(t-\vartheta)}^{\digamma(t)}\frac{1}{\varphi \digamma(t)+c\digamma^2(t)}d\digamma=-\int_{t-\vartheta}^{t}dt,
\end{array}
\]
and hence
\[\digamma(t)=\frac{\varphi \digamma(t-\vartheta)}{\varphi e^{\varphi \vartheta}+c (e^{\varphi \vartheta}-1)\digamma(t-\vartheta)}.\eqno{(A3)}\]

\newpage
\begin{figure}[]\vspace{-0.5cm}
\centering\subfloat[]{\includegraphics[width=7.1in,height=4.2in]{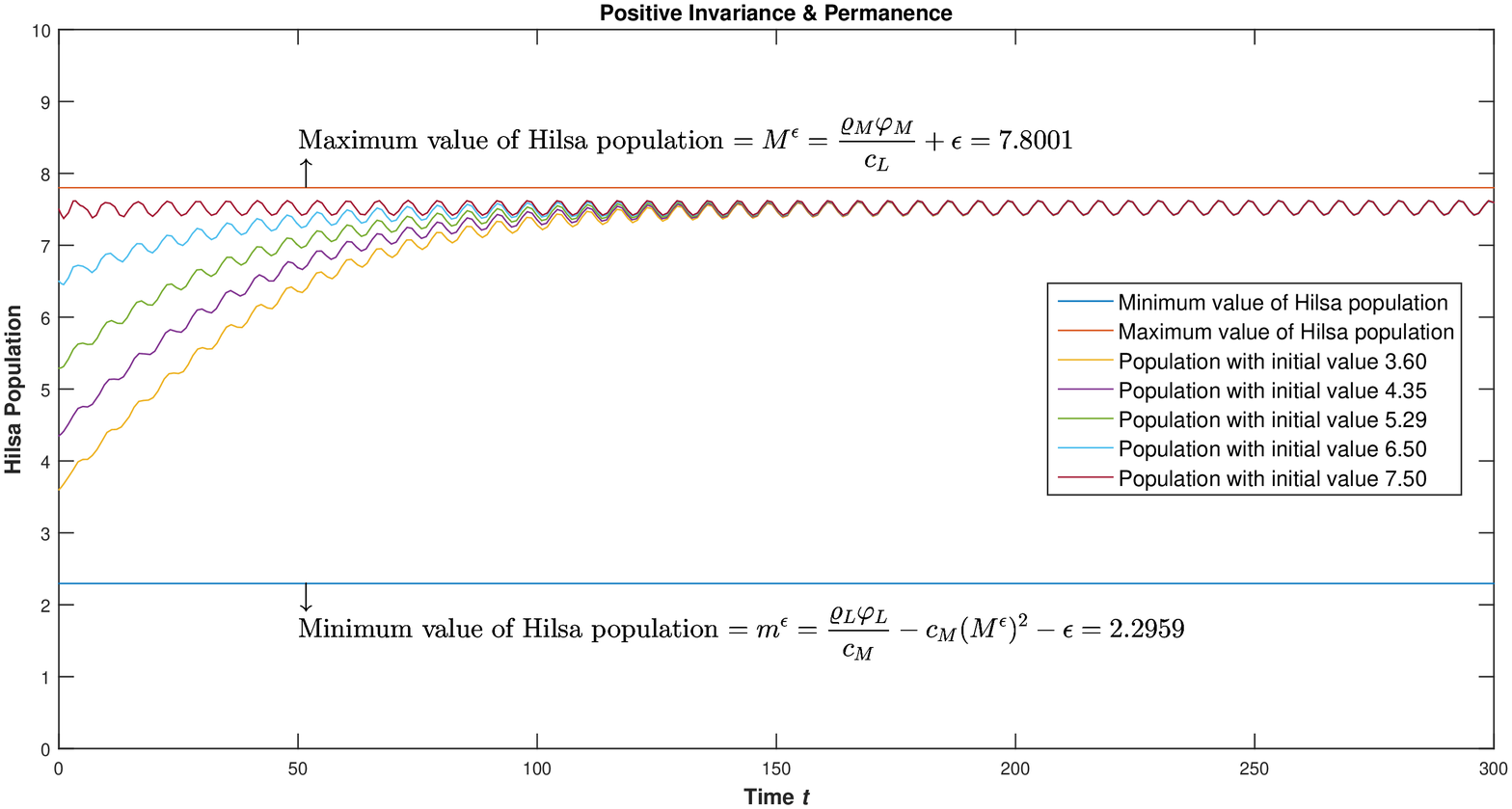}}\vspace{-0.8cm}\caption{Time series for Hilsa population $\digamma(t)$ showing positive invariance and permanence.}\label{fig:1}
\end{figure} 
\begin{figure}[]\centering
%\vspace{-0.5cm}%
\subfloat[]{\includegraphics[width=7.1in,height=4.2in]{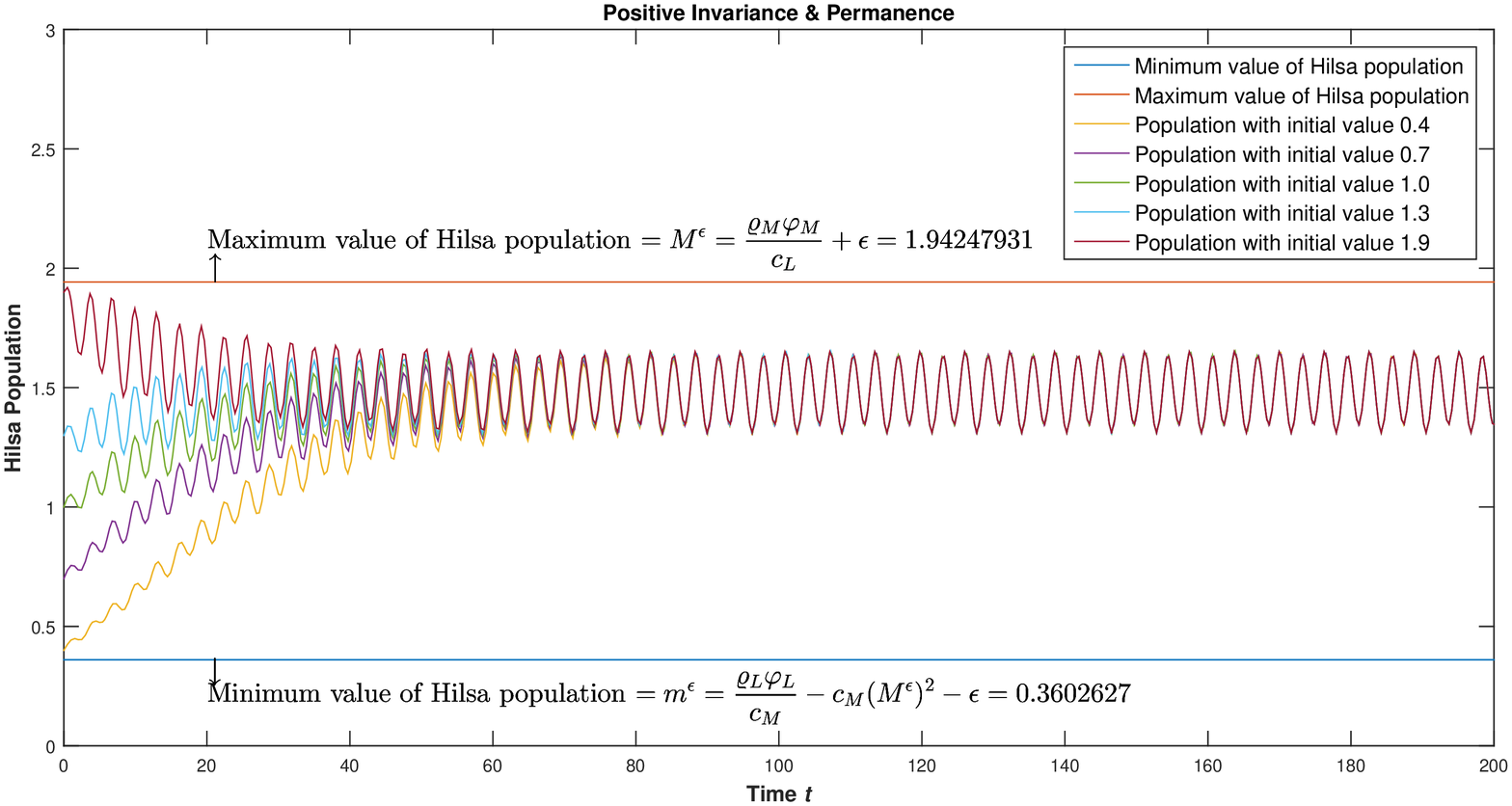}}\vspace{-0.8cm}\caption{Time series for Hilsa population $\digamma(t)$ showing positive invariance and permanence.}\label{fig:2}
\end{figure}
\begin{figure}\centering
\subfloat[]{\includegraphics[width=7.1in,height=4.2in]{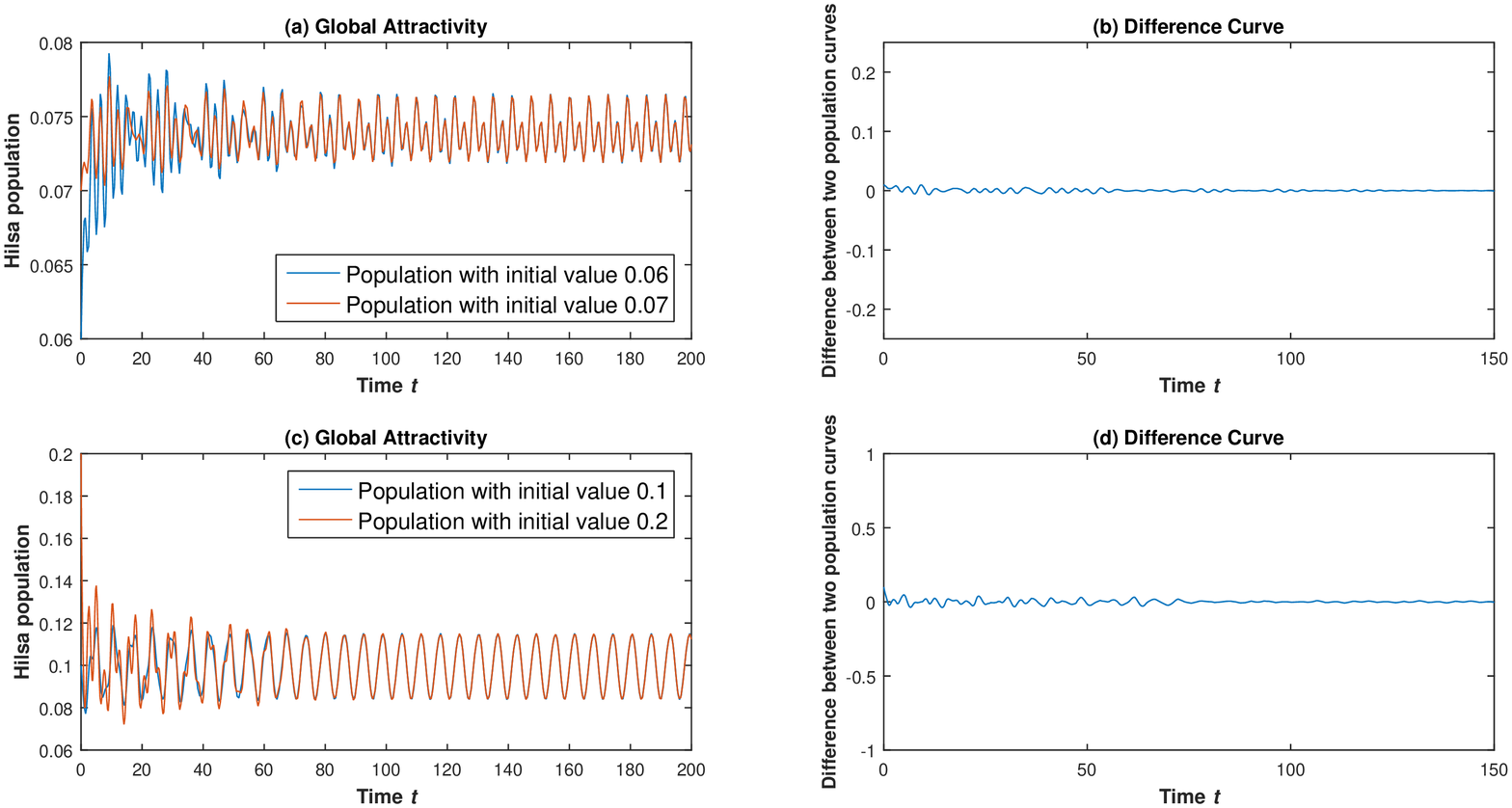}}\vspace{-0.8cm}\caption{Time series plots for population $\digamma(t)$ being globally attractive in (a) and (c). The corresponding difference curves in (b) and (c).}\label{fig:3}
\end{figure}
\begin{figure}[]\centering
\vspace{-0.8cm}\subfloat[]{\includegraphics[width=7.1in,height=4.2in]{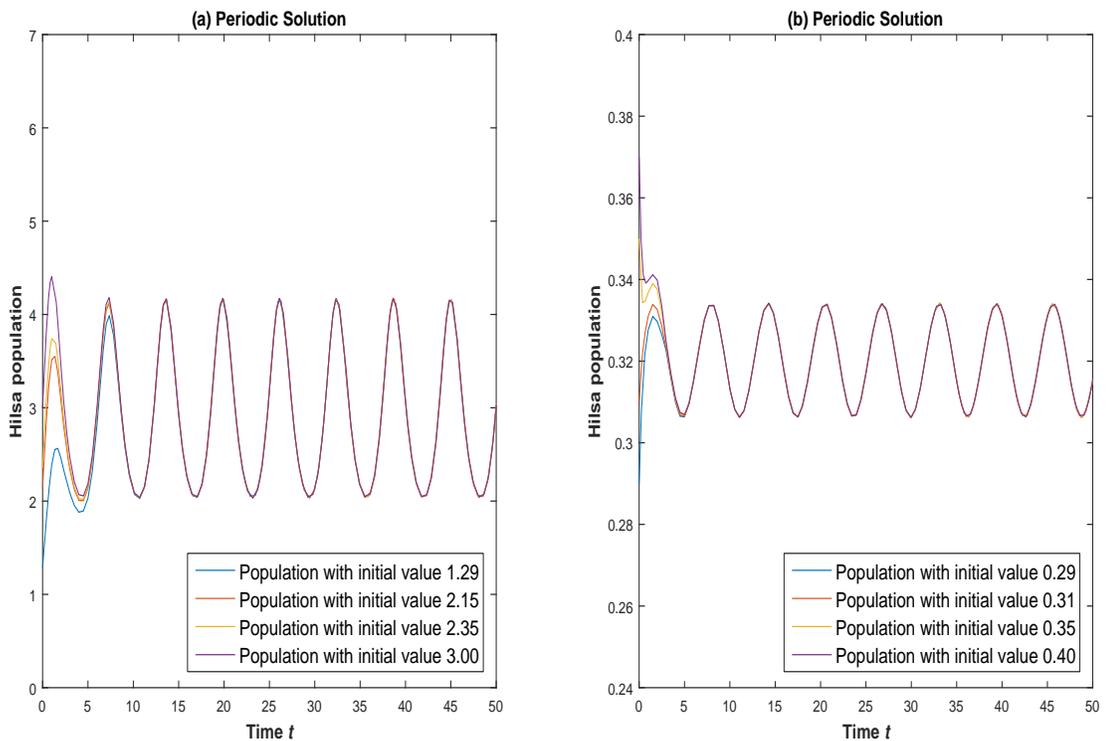}}\vspace{-0.8cm}\caption{Time series graphs showing periodicity in population $\digamma(t)$.}\label{fig:4}
\end{figure}
\begin{figure}[]\centering
\subfloat[]{\includegraphics[width=7.1in,height=4.2in]{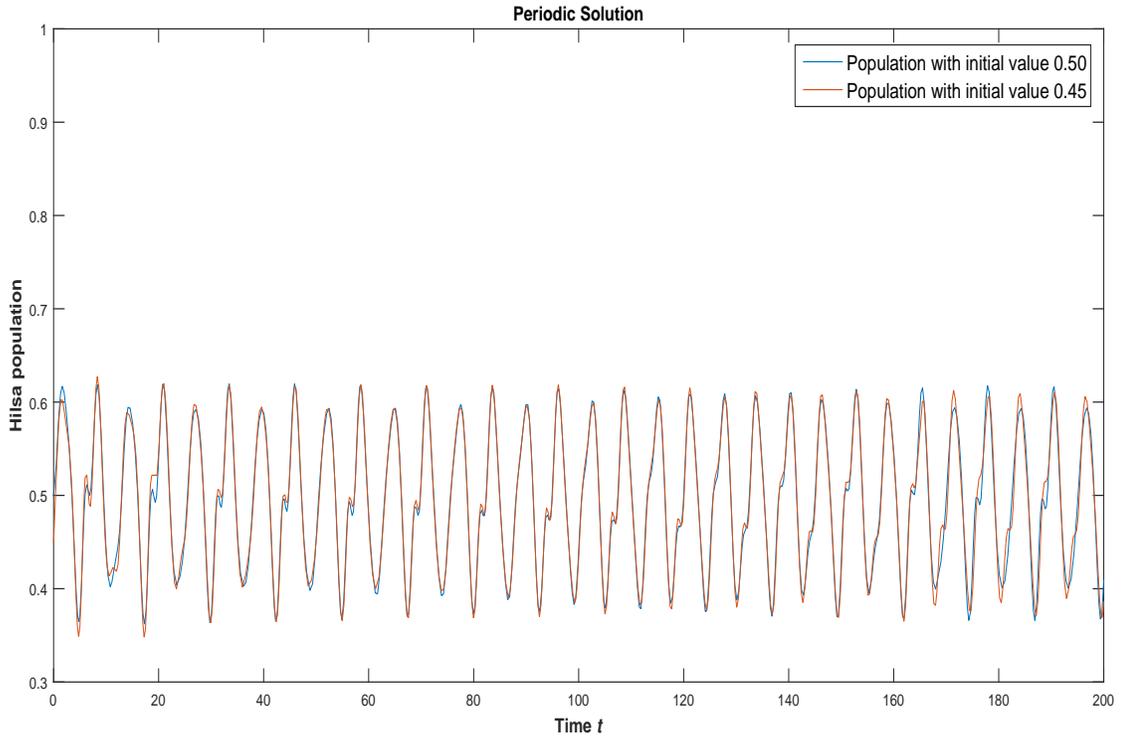}}\vspace{-0.8cm}\caption{Time series diagram for population $\digamma(t)$ being periodic in nature.}\label{fig:5}
\end{figure}
\begin{figure}\centering\vspace{-0.5cm}
\subfloat[]{\includegraphics[width=7.1in,height=4.2in]{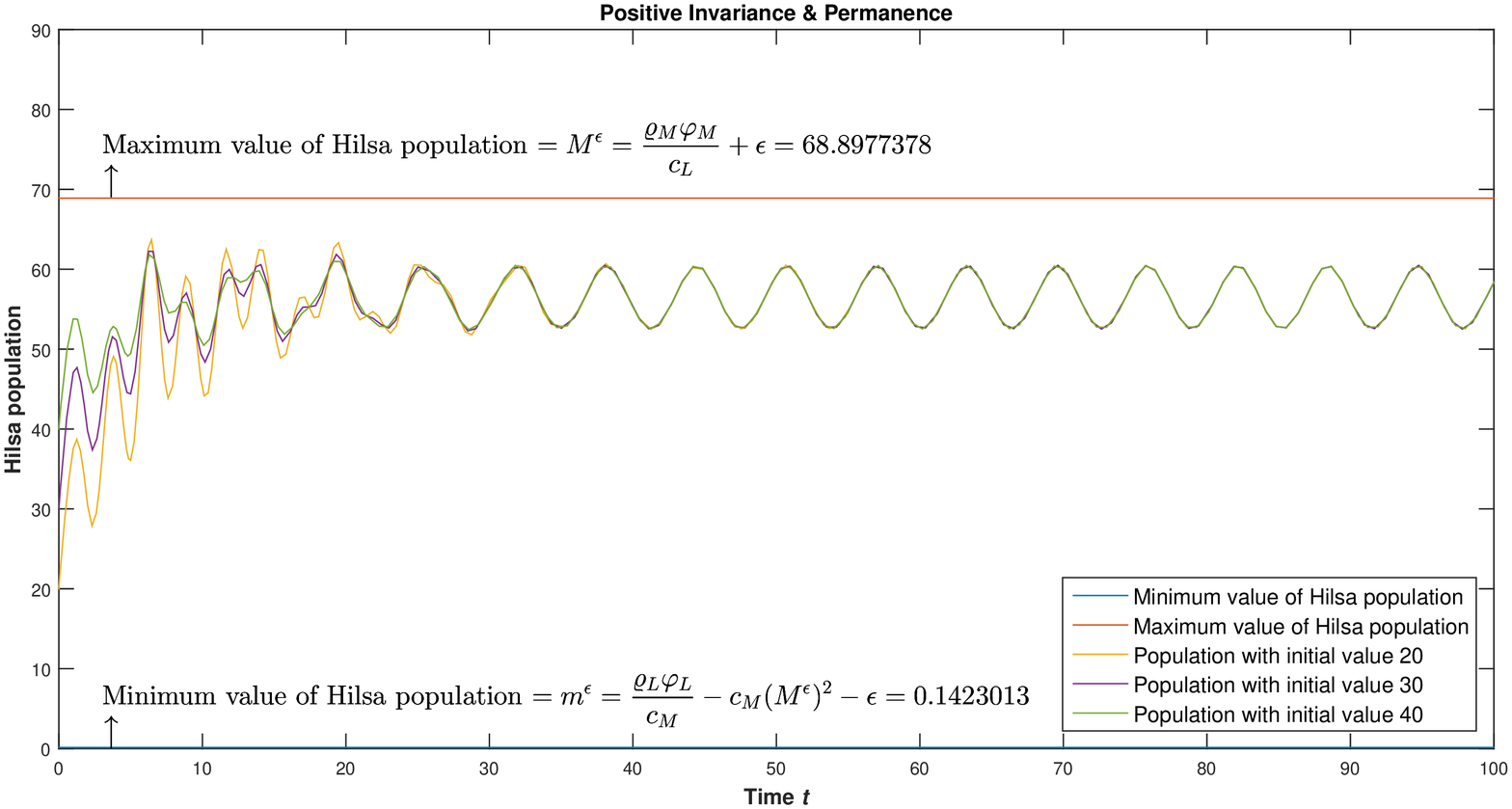}}\vspace{-0.8cm}\caption{Time series graph for Hilsa population $\digamma(t)$ showing positive invariance and permanence.}\label{fig:6}\vspace{0.2cm}
\end{figure}
\begin{figure}
\centering
\subfloat[]{\includegraphics[width=7.1in,height=4.2in]{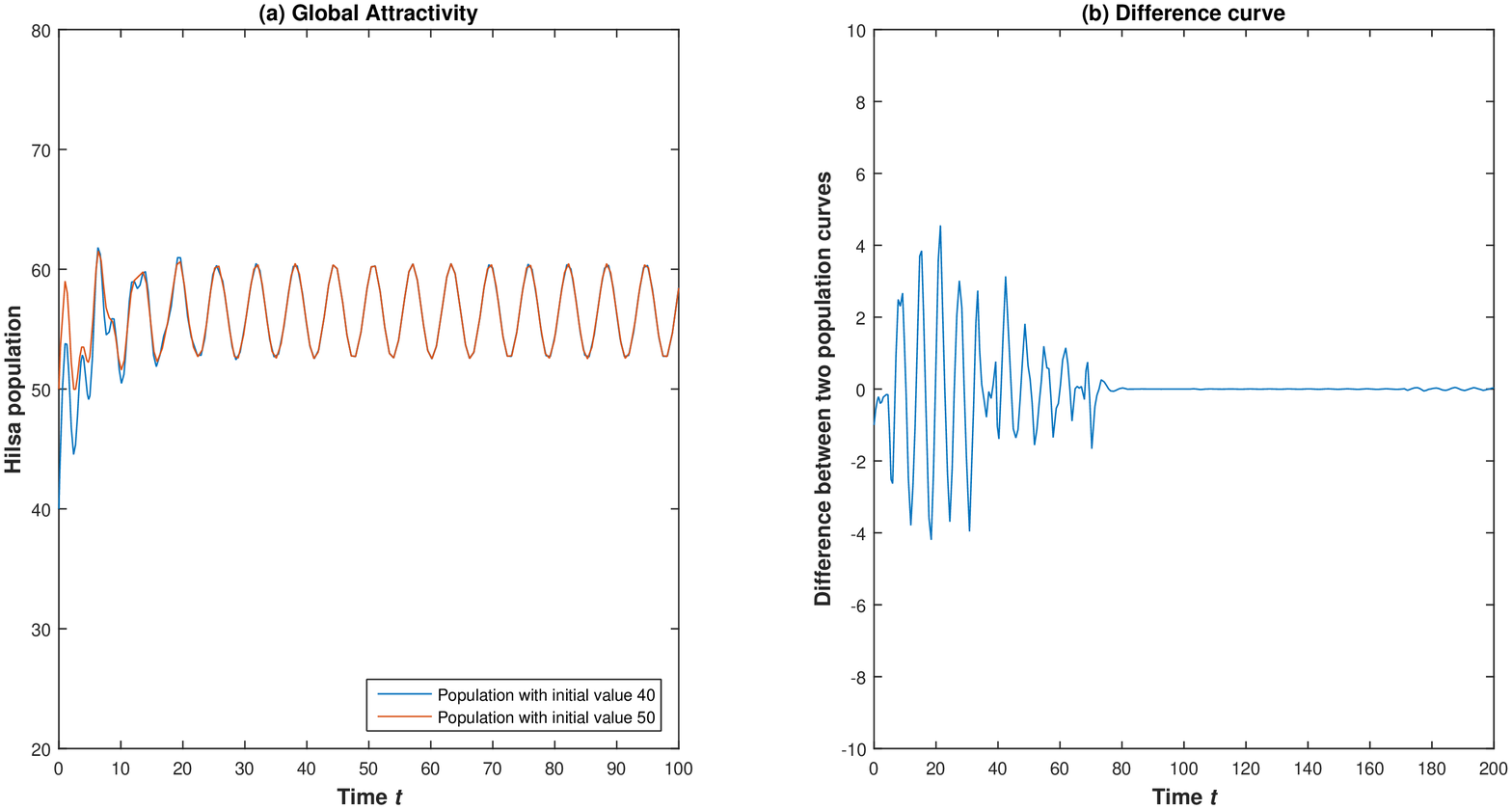}}\vspace{-0.8cm}\caption{Time series for population $\digamma(t)$ being globally attractive in (a). The corresponding difference curve in (b).}\label{fig:7}
\end{figure}
\begin{figure}[]
\vspace{-0.8cm}\subfloat[]{\includegraphics[width=7.1in,height=4.2in]{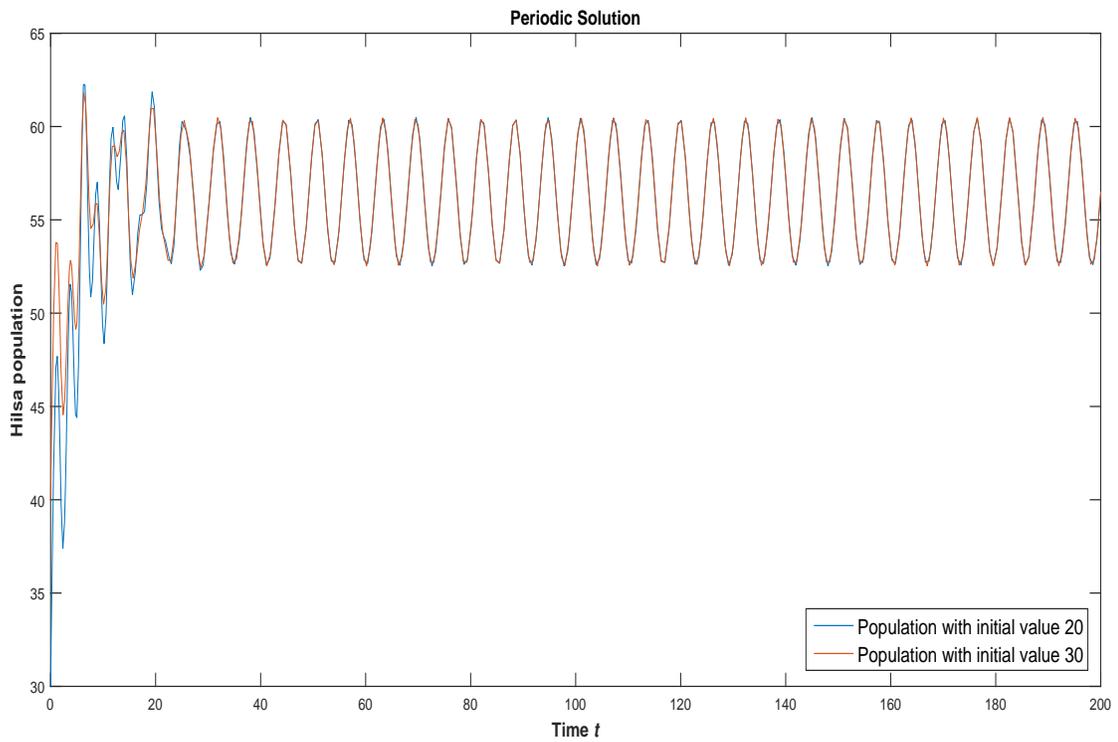}}\vspace{-0.8cm}\caption{Time series plot of a periodic solution $\digamma(t)$.}\label{fig:8}
\vspace{0.25cm}
\end{figure}
\begin{figure}[]
\vspace{-0.8cm}\subfloat[]{\includegraphics[width=7.1in,height=4.2in]{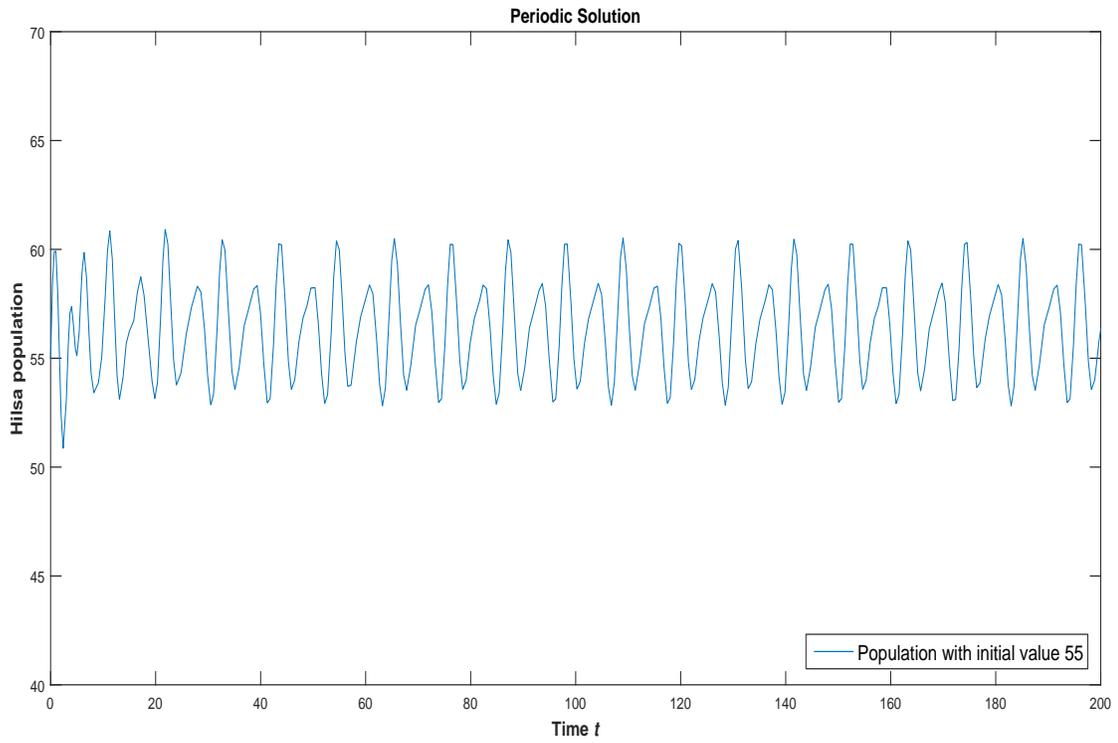}}\vspace{-0.8cm}\caption{Time series of a double periodic solution.}\label{fig:9}
\vspace{0.25cm}
\end{figure}
\begin{figure}[]
\vspace{-0.8cm}\subfloat[]{\includegraphics[width=7.1in,height=4.2in]{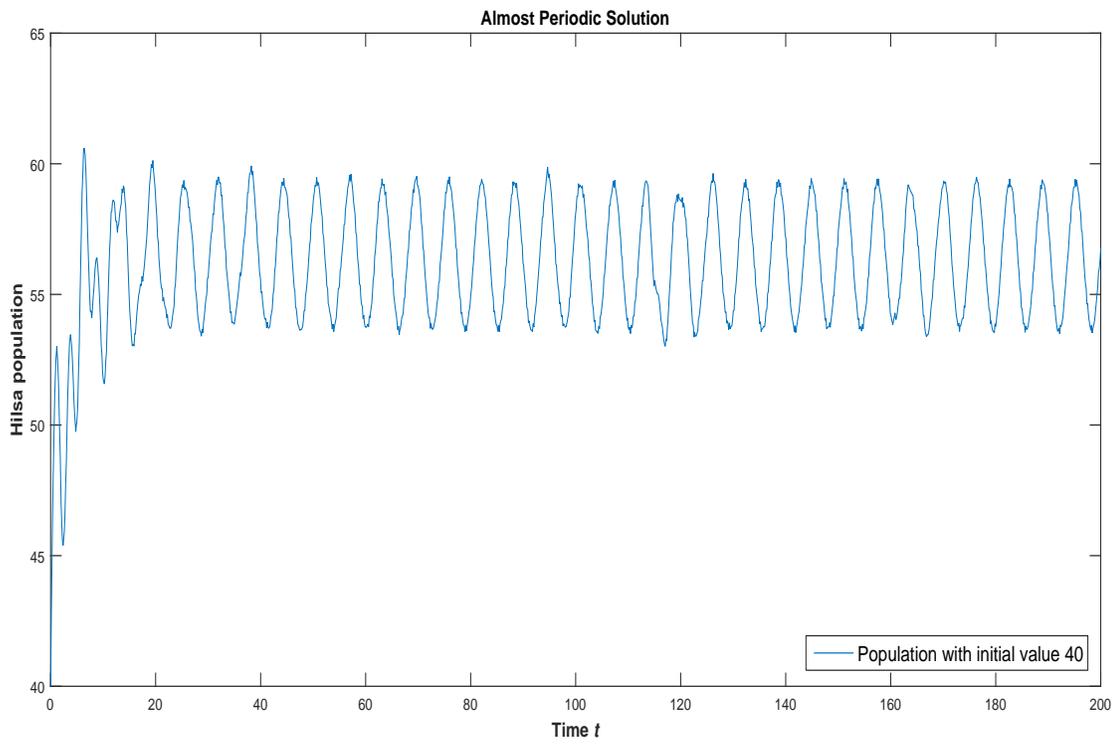}}\vspace{-0.8cm}\caption{Time series of an almost periodic solution.}\label{fig:10}
\vspace{0.25cm}
\end{figure}
\begin{figure}[]
\vspace{-0.3cm}\subfloat[]{\includegraphics[width=7.1in,height=4.2in]{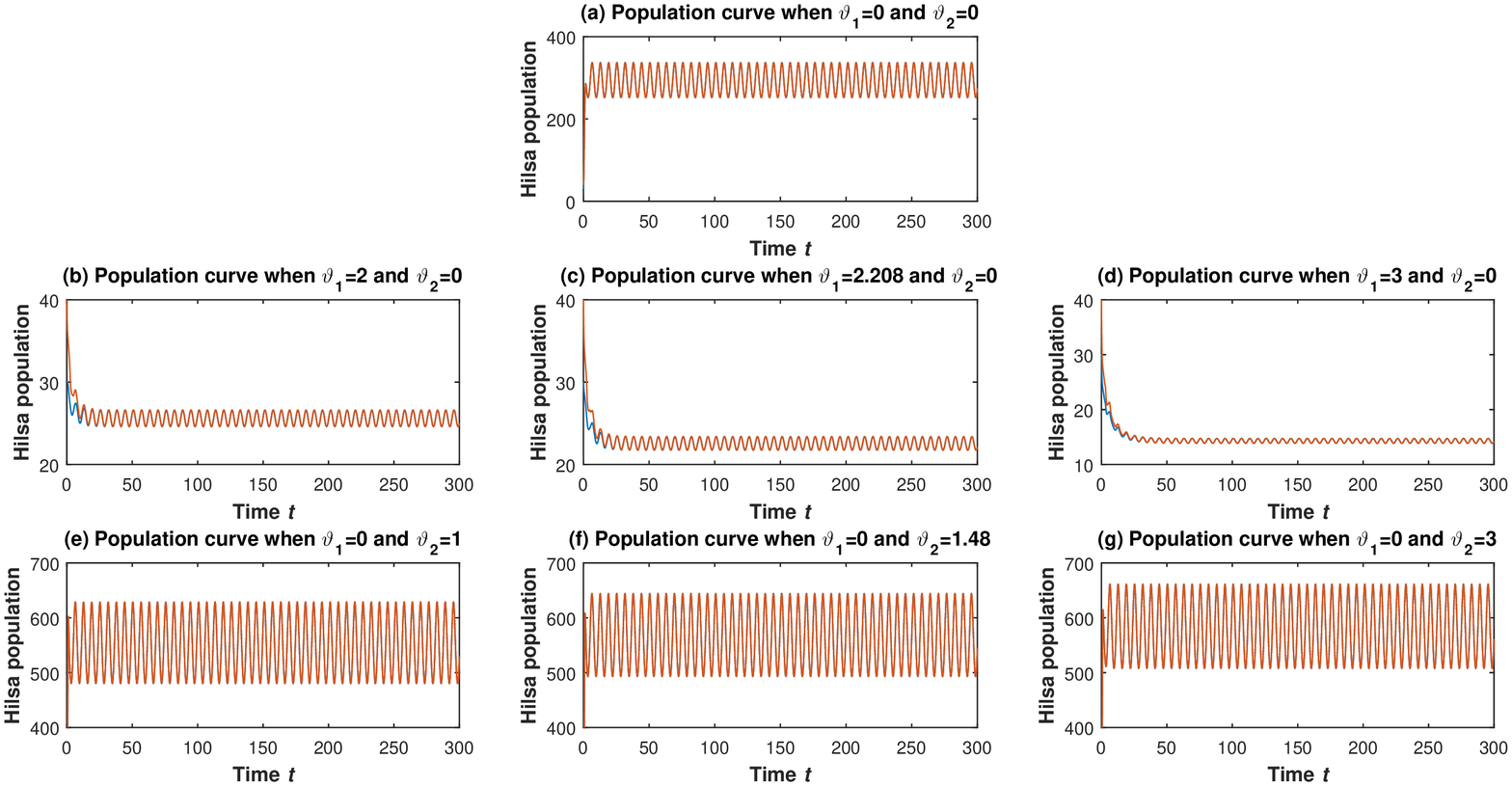}}\vspace{-0.8cm}\caption{Time series plots for Hilsa population $\digamma(t)$ when $\vartheta_1=\vartheta_2=0$ in (a); $\vartheta_1>0$, $\vartheta_2=0$ in (b),(c) and (d); $\vartheta_1=0$, $\vartheta_2>0$ in (e), (f) and (g). }\label{fig:11}
%\vspace{0.15cm}%
\end{figure}
\begin{figure}
\vspace{-0.4cm}\subfloat[]{\includegraphics[width=7.1in,height=4.2in]{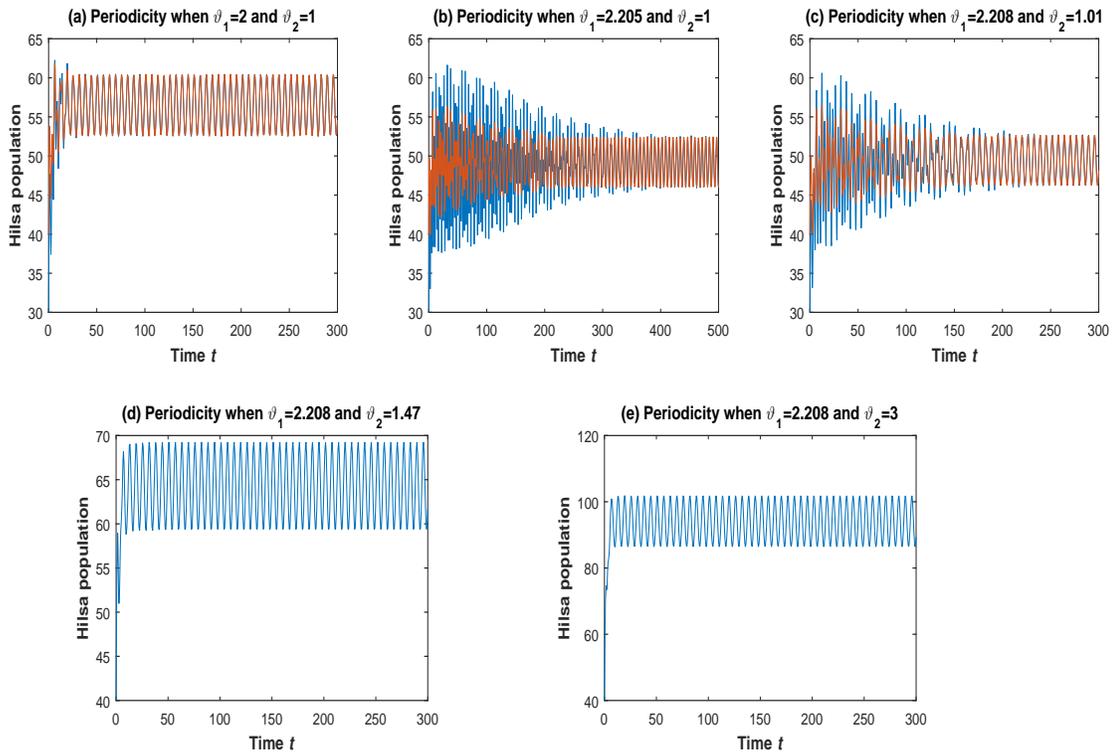}}\vspace{-0.8cm}\caption{Time series diagrams for Hilsa population $\digamma(t)$ when $\vartheta_1>0$, $\vartheta_2>0$.}\label{fig:12}
\vspace{-0.15cm}
\end{figure}
\begin{figure}[]
\vspace{-0.8cm}\subfloat[]{\includegraphics[width=7.1in,height=4.2in]{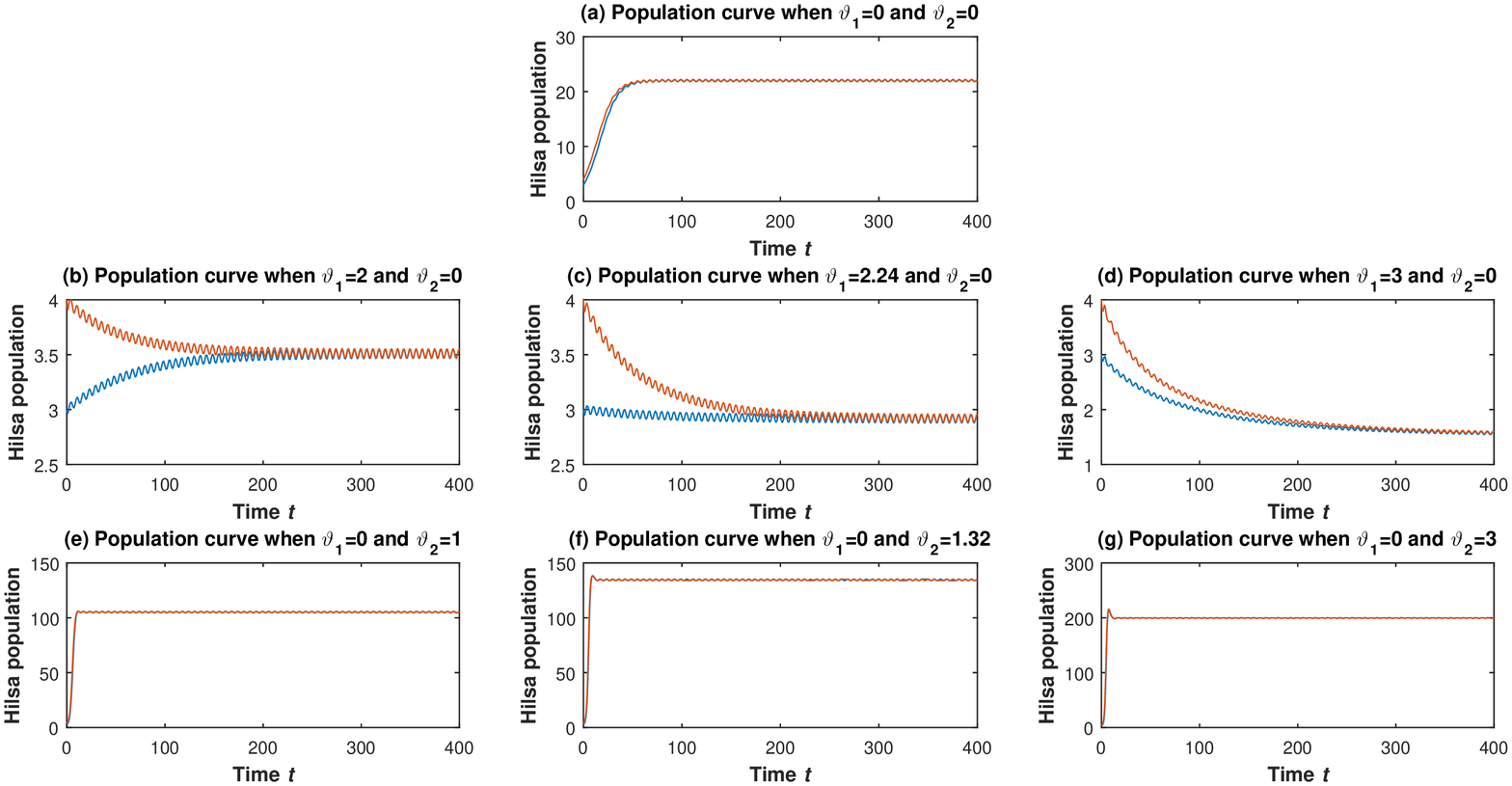}}\vspace{-0.8cm}\caption{Time plots for Hilsa population $\digamma(t)$ when $\vartheta_1=\vartheta_2=0$ in (a); $\vartheta_1>0$, $\vartheta_2=0$ in (b),(c) and (d); $\vartheta_1=0$, $\vartheta_2>0$ in (e), (f) and (g).}\label{fig:13}
\end{figure}
\begin{figure}[]
\vspace{-0.8cm}\subfloat[]{\includegraphics[width=7.1in,height=4.2in]{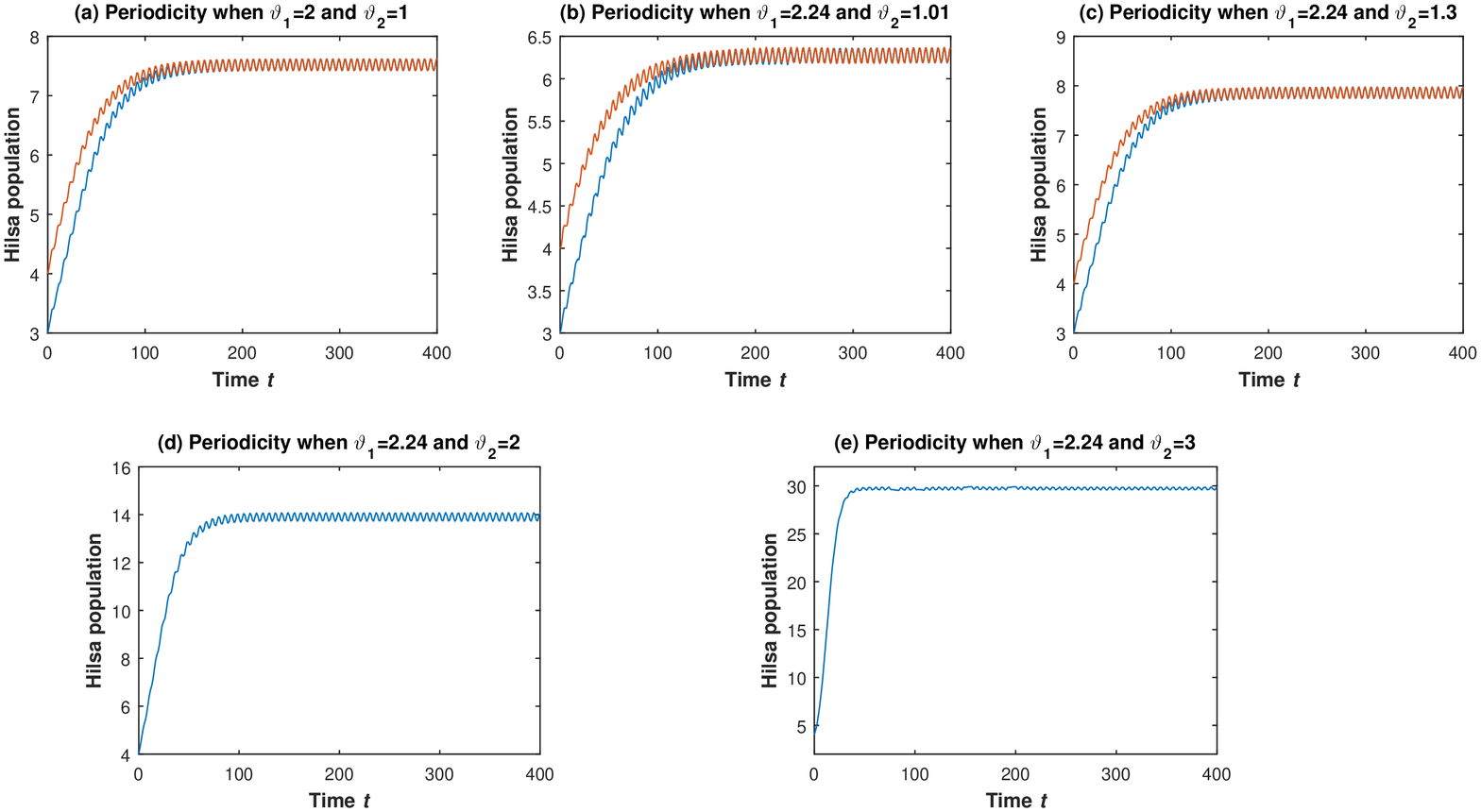}}\vspace{-0.8cm}\caption{Time series plots for Hilsa population $\digamma(t)$ when $\vartheta_1>0$, $\vartheta_2>0$.}\label{fig:14}
\end{figure}
\begin{table}[h]
\centering
%\resizebox{\columnwidth}{1}%
\caption{Population dynamics for various values of $\vartheta_1$ and $\vartheta_2$}
 %(based on observations made in Figures \ref{fig:11} and \ref{fig:12}) %
\label{Table:1}
\scalebox{0.85}{
\begin{tabular}{|p{3.2cm}|p{2.2cm}|p{3cm}|p{2.5cm}|p{2.5cm}|p{3.7cm}|} \hline 
$\mathbf{\vartheta_1}$ and $\mathbf{\vartheta_2}$ & \textbf{Positive} \textbf{Invariance} & \textbf{Permanance} & \textbf{Global} \hspace{1cm} \textbf{Attractivity}  & \textbf{Periodicity} & \textbf{Figure} \\ \hline
$\vartheta_1=0$  & \text{No } & \text{No} & \text{Yes} & \text{Yes} & \text{Fig. \ref{fig:11} (a)} \\ 
$\vartheta_2=0$ & & & & & \\ \hline
$\vartheta_1>0$ & \text{Yes} & \text{Yes} & \text{Yes} & \text{Yes}  & \text{Fig. \ref{fig:11} (b),(c),(d)}\\ 
$\vartheta_2=0$ & & & & & \\ \hline
$\vartheta_1=0$ & \text{No} & \text{No} & \text{Yes} & \text{Yes}  & \text{Fig. \ref{fig:11} (e),(f),(g)}\\ 
$\vartheta_2>0$ & & & & & \\ \hline
$2\leqslant\vartheta_1<2.208$ & \text{Yes} & \text{Yes} & \text{Yes} & \text{Yes}  & \text{Fig. \ref{fig:12} (a),(b)}\\ 
$\vartheta_2=1$ & & & & &\\ \hline
$\vartheta_1=2.208$ & \text{No} & \text{No} & \text{No} & \text{No}  & \text{}\\
$\vartheta_2=1$ & & & & &\\ \hline
$2.208\leqslant\vartheta_1\leqslant 3$ & \text{Yes} & \text{Yes} & \text{Yes} & \text{Yes}  & \text{Fig. \ref{fig:12} (c),(d)}\\
$1.01\leqslant\vartheta_2\leqslant 1.47$ & & & & &\\ \hline
$2.208\leqslant\vartheta_1\leqslant 3$  & \text{No} & \text{No} & \text{Yes} & \text{Yes}  & \text{Fig. \ref{fig:12} (e)}\\ 
$1.47<\vartheta_2\leqslant 3$ & & & & &\\ \hline
\end{tabular}}
\end{table}
\begin{table}[h]
\centering
\caption{Population based on various values of $\vartheta_1$ and $\vartheta_2$}
%(see Figures \ref{fig:11} and \ref{fig:12})%
\label{Table:2}\scalebox{0.98}{
\begin{tabular}{|p{3cm}|p{13cm}|} \hline
$\mathbf{\vartheta_1}$ and $\mathbf{\vartheta_2}$ & \textbf{Nature of the system}\\ \hline
$\vartheta_1=0$  & \text{When there is no delay in maturity and harvestation the population }\\ 
$\vartheta_2=0$ & \text{grows immensely. (Fig. \ref{fig:11} (a))}\\ \hline
$\vartheta_1>0$ & \text{When there is no harvesting delay the population decreases.}\\ 
$\vartheta_2=0$ &  \text{ (Fig. \ref{fig:11} (b),(c),(d))} \\ \hline
$\vartheta_1=0$ & \text{When there is no delay in maturity the population increases enormously.}\\ 
$\vartheta_2>0$ &  \text{ (Fig. \ref{fig:11} (e),(f),(g))}\\ \hline
$2\leqslant\vartheta_1<2.208$ & \text{Population is positively invariant, permanent, globally stable and }\\ 
$\vartheta_2=1$ &  \text{periodic. (Fig. \ref{fig:12} (a),(b))}\\ \hline
$\vartheta_1=2.208$ & \text{Population curve(s) become unbounded i.e, It is not positively invariant}\\ 
$\vartheta_2=1$ &  \text{nor permanent nor globally stable nor periodic.}\\ \hline
$2.208\leqslant\vartheta_1\leqslant 3$ & \text{Population is positively invariant, permanent, globally stable and}\\
$1.01\leqslant\vartheta_2\leqslant 1.47$ & \text{periodic. (Fig. \ref{fig:12} (c),(d))} \\ \hline
$2.208\leqslant\vartheta_1\leqslant 3$  & \text{Population increases rapidly hence it stops being positively invariant}\\ 
$1.47<\vartheta_2\leqslant 3$ & \text{and permanent. But it stays globally attractive and periodic. } \\
& \text{(Fig. \ref{fig:12} (e))}\\ \hline
\end{tabular}}
\end{table} 
\begin{table}[h]
\centering
\caption{Population dynamics for various values of $\vartheta_1$ and $\vartheta_2$}
% (based on observations made in Figures \ref{fig:13} and \ref{fig:14})}%
\label{Table:3}
\scalebox{0.85}{
\begin{tabular}{|p{3cm}|p{2.2cm}|p{3cm}|p{2.5cm}|p{2.5cm}|p{4cm}|} \hline
$\mathbf{\vartheta_1}$ and $\mathbf{\vartheta_2}$ & \textbf{Positive} \textbf{Invariance} & \textbf{Permanance} & \textbf{Global} \hspace{1cm} \textbf{Attractivity}  & \textbf{Periodicity} & \textbf{Figure} \\ \hline
$\vartheta_1=0$  & \text{No} & \text{No} & \text{Yes} & \text{Yes} & \text{Fig. \ref{fig:13} (a)}\\ 
$\vartheta_2=0$ & & & & &\\ \hline
$\vartheta_1>0$ & \text{Yes} & \text{Yes} & \text{Yes} & \text{Yes} & \text{Fig. \ref{fig:13} (b),(c),(d)}\\
$\vartheta_2=0$ & & & && \\ \hline
$\vartheta_1=0$ & \text{No} & \text{No} & \text{Yes} & \text{Yes} & \text{Fig. \ref{fig:13} (e),(f),(g)}\\ 
$\vartheta_2>0$ & & & & &\\ \hline
$2\leqslant\vartheta_1<2.24$ & \text{Yes} & \text{Yes} & \text{Yes} & \text{Yes} & \text{Fig. \ref{fig:14} (a)}\\ 
$\vartheta_2=1$ & & & & &\\ \hline
$\vartheta_1=2.24$ & \text{No} & \text{No} & \text{No} & \text{No} & \text{}\\ 
$\vartheta_2=1$ & & & && \\ \hline
$2.24\leqslant\vartheta_1\leqslant 3$ & \text{Yes} & \text{Yes} & \text{Yes} & \text{Yes} & \text{Fig. \ref{fig:14} (b),(c)}\\
$1.01\leqslant\vartheta_2\leqslant 1.3$ & & & & &\\ \hline
$2.24\leqslant\vartheta_1\leqslant 3$  & \text{No} & \text{No} & \text{Yes} & \text{Yes} & \text{Fig. \ref{fig:14} (d),(e)}\\ 
$1.3<\vartheta_2\leqslant 3$ & & & & &\\ \hline
\end{tabular}}
\end{table}
\begin{table}[h]
\centering
\caption{Population based on various values of $\vartheta_1$ and $\vartheta_2$}
% (see Figures \ref{fig:13} and \ref{fig:14})}%
\label{Table:4}
\begin{tabular}{|p{3cm}|p{13cm}|}\hline
$\mathbf{\vartheta_1}$ and $\mathbf{\vartheta_2}$ & \textbf{Nature of the system}\\ \hline
$\vartheta_1=0$ & \text{When there is no delay in maturity and harvestation the population }\\ 
$\vartheta_2=0$ & \text{grows immensely. (Fig. \ref{fig:13} (a))}\\ \hline
$\vartheta_1>0$ & \text{When there is no harvesting delay the population decreases.}\\ 
$\vartheta_2=0$ & \text{ (Fig. \ref{fig:13} (b),(c),(d))}\\ \hline
$\vartheta_1=0$ & \text{When there is no delay in maturity the population increases enormously. }\\ 
$\vartheta_2>0$ & \text{(Fig. \ref{fig:13} (e),(f),(g))} \\ \hline
$2\leqslant\vartheta_1<2.24$ & \text{Population is positively invariant, permanent, globally stable and }\\ 
$\vartheta_2=1$ &  \text{periodic. (Fig. \ref{fig:14} (a))}\\ \hline
$\vartheta_1=2.24$ & \text{Population curve(s) become unbounded i.e, It is not positively invariant}\\ 
$\vartheta_2=1$ &  \text{nor permanent nor globally stable nor periodic.}\\ \hline
$2.24\leqslant\vartheta_1\leqslant 3$ & \text{Population is positively invariant, permanent, globally stable and}\\
$1.01\leqslant\vartheta_2\leqslant 1.3$ & \text{periodic. (Fig. \ref{fig:14} (b),(c))} \\ \hline
$2.24\leqslant\vartheta_1\leqslant 3$  & \text{Population increases rapidly hence it stops being positively invariant}\\ 
$1.3<\vartheta_2\leqslant 3$ & \text{and permanent. But it stays globally attractive and periodic. } \\ 
& \text{(Fig. \ref{fig:14} (d),(e))}\\ \hline
\end{tabular}
\end{table}

\end{document}